\documentclass[11pt]{article}

\usepackage[bookmarks,colorlinks,breaklinks]{hyperref}  
\hypersetup{linkcolor=blue,citecolor=blue,filecolor=blue,urlcolor=blue} 
\usepackage{fullpage,appendix}
\usepackage{amsmath,amsfonts,amsthm,amssymb,xspace,bm}
\usepackage[letterpaper,
top=1in,
bottom=1in,
left=1in,
right=1in]{geometry}

\newcommand{\newref}[2][]{\hyperref[#2]{#1~\ref*{#2}}}
\renewcommand{\eqref}[1]{\hyperref[#1]{(\ref*{#1})}}
\numberwithin{equation}{section}

\newcommand{\sref}[1]{\newref[Section]{#1}}

\newcommand{\tref}[1]{\newref[Theorem]{#1}}
\newcommand{\lref}[1]{\newref[Lemma]{#1}}

\newcommand{\cref}[1]{\newref[Corollary]{#1}}

\newcommand{\eref}[1]{\hyperref[1]{Equation~\eqref{#1}}}

\newcommand{\clref}[1]{\newref[Claim]{#1}}

\newcommand{\anote}[2]{}

\newcommand{\Snote}[1]{\anote{Salil}{#1}}

\newcommand{\Mnote}[1]{\anote{Raghu}{#1}}

\theoremstyle{plain}
\newtheorem{theorem}{Theorem}[section]
\newtheorem{Thm}[theorem]{Theorem}
\newtheorem{Lem}[theorem]{Lemma}
\newtheorem{Claim}[theorem]{Claim}

\newtheorem{Cor}[theorem]{Corollary}

\newtheorem{Def}[theorem]{Definition}

\theoremstyle{definition}

\DeclareMathOperator*{\Var}{\mathsf{Var}}

\DeclareMathOperator*{\E}{\mathbb{E}}
\DeclareMathOperator*{\Ll}{\mathsf{L_1}}
\DeclareMathOperator*{\pr}{\mathbb{P}}

\def\inpw#1,#2{\langle #1, #2\rangle}

\newcommand{\Path}{\mathsf{Path}}
\newcommand{\rgta}{\rightarrow}
\newcommand{\lfta}{\leftarrow}

\newcommand{\eat}[1]{}
\newcommand{\etal}{et al.}

\newcommand{\reals}{\mathbb{R}}

\newcommand{\dpm}{\{ \pm 1\}}

\newcommand{\eps}{\epsilon}

\newcommand{\note}[1]{\marginpar{\tiny *note in TeX*}}
\newcommand{\ignore}[1]{}

\newcommand{\calD}{{\cal D}}

\newcommand{\tO}{\tilde{O}}
\renewcommand{\phi}{\varphi}
\renewcommand{\epsilon}{\varepsilon}

\newcommand{\pmo}{\{ \pm 1\}}

\newcommand{\R}{\mathbb{R}}

\newcommand{\poly}{\mathrm{poly}}
\newcommand{\polylog}{\mathrm{polylog}}

\newcommand{\Ands}[1]{\ensuremath{\mathrm{And}(#1)}}
\newcommand{\Ors}[1]{\ensuremath{\mathrm{Or}(#1)}}
\newcommand{\Xors}[1]{\ensuremath{\mathrm{XOR}(#1)}}
\newcommand{\Xor}{\ensuremath{\mathrm{XOR}}}
\newcommand{\Pk}{P_{\leq k}}

\newcommand{\mc}[1]{\ensuremath{\mathcal{#1}}}

\newcommand{\ind}[1]{\mathbf{1}_{#1}}

\newcommand{\C}{{{\cal C}}}
\newcommand{\D}{{{\cal D}}}

\newcommand{\cD}{{\ensuremath{x \sim \mathcal{D}}}}

\newcommand{\cU}{{\ensuremath{x \sim \dpm^n}}}
\newcommand{\yu}{{\ensuremath{y \sim \dpvm}}}
\newcommand{\xu}{{\ensuremath{\bar{x} \sim \mathcal{U}}}}
\newcommand{\zo}{\{0,1\}}

\newcommand{\Start}{{\sf Start}}
\newcommand{\Rej}{{\sf Rej}}
\newcommand{\Acc}{{\sf Acc}}
\newcommand{\Bad}{{\sf Bad}}

\newcommand{\cnfx}{\mathsf{CNF^\oplus}}
\newcommand{\cnf}{\ensuremath {\mathsf{CNF}}}
\newcommand{\dnf}{\ensuremath {\mathsf{DNF}}}
\newcommand{\prg}{\ensuremath {\mathsf{PRG}}}
\newcommand{\hsg}{\ensuremath {\mathsf{HSG}}}
\newcommand{\bp}[1]{\mathsf{BP(#1)}}
\newcommand{\bpr}[1]{\mathsf{BP^{Rej}(#1)}}
\newcommand{\rcnf}{\ensuremath {\mathsf{RCNF}}}

\title{Better Pseudorandom Generators from Milder Pseudorandom Restrictions.}
\author{
Parikshit Gopalan\\
MSR-SVC\\
\and
Raghu Meka\\
IAS Princeton\thanks{Supported in part by NSF grant DMS-0835373. Work done in part while the author was an intern at Microsoft Research Silicon Valley.}\\
\and
Omer Reingold\\
MSR-SVC\\
\and
Luca Trevisan\\
Stanford University
\and
Salil Vadhan\thanks{School of Engineering and Applied Sciences, Harvard University, Cambridge, MA 02138.  \texttt{salil@seas.harvard.edu}.  Supported in part by NSF grant CCF-1116616. Work done in part while on leave as a Visiting Researcher at Microsoft Research Silicon Valley and a Visiting Scholar at Stanford University.}  \\
Harvard University
}
\date{}
\begin{document}
\maketitle
\thispagestyle{empty}

\begin{abstract}
We present an iterative approach to constructing pseudorandom
generators, based on the repeated application of mild pseudorandom
restrictions. We use this template to construct pseudorandom
generators for combinatorial rectangles and read-once \cnf s and a
hitting set generator for width-3 branching programs, all of which achieve
near-optimal seed-length even in the low-error regime: We get
seed-length $\tilde{O}(\log (n/\epsilon))$ for error
$\epsilon$. Previously, only constructions with seed-length
$O(\log^{3/2} n)$ or $O(\log^2 n)$ were known for these classes with
error $\epsilon = 1/\poly(n)$.

The (pseudo)random restrictions we use are milder than those typically
used for proving circuit lower bounds in that we only set a constant
fraction of the bits at a time. While such restrictions do not
simplify the functions drastically, we show that they can be
derandomized using small-bias spaces.
\end{abstract}

\Snote{added thanks line for myself, please propagate to camera-ready version.  (address not important, but mention of NSF grant, MSR, and Stanford are).}
\Mnote{added thanks for me too}
\newpage
\addtocounter{page}{-1}
\newcommand{\tOmega}{\tilde{\Omega}}

\section{Introduction}
\subsection{Pseudorandom Generators}
The theory of pseudorandomness has given compelling evidence that very strong
pseudorandom generators exist.  For example, assuming that there are computational problems solvable in exponential time that require exponential-sized circuits, Impagliazzo and Wigderson~\cite{ImpagliazzoWi97} have shown that for every $n$, $c$ and
$\eps>0$, there exist efficient pseudorandom generators (\prg s) mapping a
random seed of length $O(\log(n^c/\eps))$ to $n$ pseudorandom bits that
cannot be distinguished from $n$ uniformly random bits with
probability more than $\eps$, by any Boolean circuit of size $n^c$.
These \prg s, which fool arbitrary efficient computations (represented
by polynomial-sized Boolean circuits), have remarkable consequences
for derandomization: every randomized algorithm can be made
deterministic with only a polynomial slowdown, and thus $\mathrm{P}=\mathrm{BPP}$.

These results, however, remain conditional on a circuit complexity assumption whose proof seems far off at present.  Since \prg s that fool
a class of Boolean circuits also imply lower bounds for that class, we
cannot hope to remove the assumption. Thus unconditional generators
are only possible for  restricted models of computation for which we
have lower bounds.

{\em Bounded-depth circuits} and {\em bounded-space algorithms} are two models of
computations for which we know how to construct \prg s with
$O(\log^{O(1)}( n/\epsilon))$ seed length~\cite{Nisan91,Nisan92}. Known \prg\  constructions for these classes
have found several striking applications including the design of streaming
algorithms~\cite{Indyk}, algorithmic
derandomization \cite{Sivakumar}, randomness extractors \cite{Trevisan},
hashing~\cite{CReingoldSW}, hardness
amplification~\cite{HealyVaVi}, almost $k$-wise independent
permutations \cite{KaplanNaRe}, and cryptographic \prg s
\cite{HaitnerHaRe06}. Arguably, constructing \prg s with the optimal
$O(\log(n/\epsilon))$ seed length for these classes are two of the
outstanding open problems in derandomization.

Nisan~\cite{Nisan92} devised a \prg\  of seed length $O(\log^2 n)$ that
fools polynomial-width branching programs, the non-uniform model of
computation that captures logspace randomized algorithms:
a space-$s$ algorithm is modeled by a branching program\footnote{Space-bounded
randomized algorithms are modeled by {\em oblivious, read-once} branching programs,
which read the input bits in a specified order and read each input bit
only once. In this paper, all the references to ``branching programs''
refer to ``oblivious read-once branching programs.''}
of width $2^s$. Nisan's generator has been used by Saks and
Zhou~\cite{SaksZh95} to prove that every randomized logspace
algorithms can be simulated in space $O(\log^{3/2} n)$,
Nisan's generator remains the best
known generator for polynomial-width branching programs (and logspace
randomized algorithms) and, despite  much progress in
this area \cite{ImpagliazzoNW,NisanZ,RazR,Reingold08,ReingoldTrVa06,BravermanRRY,BrodyV10,KouckyNiPu11,De11}, there are
very few cases where we can improve on Nisan's twenty year old bound
of $O(\log^2 n)$ \cite{Nisan92}.  
\Snote{added citation to works on consistently labelled graphs and permutation branching programs, as well as following two sentences}
For constant-width {\em regular} branching programs, Braverman et al.~\cite{BravermanRRY} have given a
pseudorandom generator with seed length $\tO((\log n)\cdot (\log(1/\eps)))$, which
is $\tO(\log n)$ for $\eps=1/\polylog(n)$, but is no better than Nisan's generator
when $\eps=1/\poly(n)$.  Only for constant-width {\em permutation} branching programs and for width-2 branching programs has seed length $O(\log(n/\eps))$ been achieved, by
Kouck{\'y}, Nimbhorkar, Pudl{\'a}k~\cite{KouckyNiPu11} and Saks and Zuckerman~\cite{SaksZ}, respectively.
Remarkably, even for width-3 branching programs 
we do not know of any efficiently computable \prg\  with seed length $o(\log^2 n)$.
Recently, Sima and Zak~\cite{SimaZa11} have constructed {\em hitting set
  generators} (\hsg s, which are a weaker form of pseudorandom generators) for width-3 branching programs with optimal seed length $O(\log
  n)$, for a large error parameter $\epsilon > 5/6$.

In a different work, Nisan~\cite{Nisan91} also gave a gives a \prg\  that $\epsilon$-fools  AC$_0$
circuits of depth $d$ and size $s$ using seed length $O(\log^{2d+6}
(s/\epsilon))$.  \Snote{changed citation from NW to Nisan91}
For the special case of depth-2 circuits, that is, \cnf s and \dnf s, the work of Bazzi \cite{Bazzi}, simplified by Razborov \cite{Razborov},
provides a \prg\  of seed length $O(\log n \cdot \log^2 (s/\epsilon))$, which has been improved to $\tilde O(\log^2 (s/\epsilon))$ by De et al.~\cite{DeEtTrTu10}.
For the restricted case of {\em read-$k$} \dnf s and \cnf s, De et al. (for
$k$ =1), and Klivans et al.~\cite{KlivansLW10} (for $k$ constant)
improve the seed length to $O(\log \epsilon^{-1} \cdot \log s)$, which
is  optimal for constant  $\epsilon$, but it is essentially no better
than the bound for general \cnf s and \dnf s when $\epsilon$ is polynomial
in $1/n$.

The model of {\em combinatorial rectangles} is closely related to both bounded-width branching programs and read-once \cnf s and are interesting combinatorial objects with a variety of applications of their own \cite{ArmoniSWZ}. The problem of constructing \prg s for combinatorial rectangles is closely related to the construction of small sample spaces that approximate the uniform distribution on many multivalued random variables \cite{EvenGLNV}: they can be seen as an alternate generalization of the versatile notion of almost $k$-wise independent distributions on $\zo^n$ to 
larger domains $[m]^n$. \Snote{removed phrase almost $k$-wise independent distribution from the larger-domain case.  I think nowadays almost $k$-wise independence typically means that the statistical distance of any $k$ coordinates is close to uniform}
Versions of this problem  where each coordinate is a real interval were first studied in number theory and analysis \cite{ArmoniSWZ}. Subsequently there has been much work on this problem \cite{EvenGLNV,LinialLSZ,ArmoniSWZ,Lu,Viola11b}. A \prg\ with seed length $O(\log
  n+\log^{3/2}(1/\eps))$~\cite{Lu} is known for combinatorial
  rectangles; such a generator achieves the optimal seed length
  $O(\log n)$ when $\epsilon \geq 2^{-O(\log^{2/3} n)}$, but not for
  $\epsilon=1/\poly(n)$. It is known how to construct \hsg s
  (which are a weakening of \prg s) with seed length
  $O(\log(n/\eps))$~\cite{LinialLSZ}.

Indeed, there are few models of computations for which we know how to construct
\prg s with the optimal seed length $O(\log(n/\eps))$ or even $\log^{1+o(1)}(n/\eps)$. The
most prominent examples are bounded-degree polynomials over finite fields~\cite{NaorNa,AlonGoHaPe,BogdanovVi,Lovett,Viola}, with parities (which are fooled by small-bias distributions~\cite{NaorNa}) as a special case, and models that can be reduced to these cases, such as width-2 branching programs \cite{SaksZ,BogdanovDvVeYe09}. 

In summary, there are several interesting models of computation
for which a {\em polylogarithmic} dependence on $n$ and $1/\eps$ is
known, and the dependence on one parameter is logarithmic on its own
(e.g. seed length $O(\log n \log(1/\eps))$), but a logarithmic bound
in both parameters together has been elusive. Finally, we remark that
not having a logarithmic dependence on the error $\epsilon$ is often a
symptom of a more fundamental bottleneck. For instance, \hsg s with constant error for width $4$ branching programs imply \hsg s with polynomially small error
for width $3$ branching programs, so achieving the latter is a natural first step towards the former.  A polynomial-time
computable \prg\  for \cnf s with seed length $O(\log
n/\epsilon)$ would imply the existence of a problem in exponential time that
requires depth-3 circuits of size $2^{\Omega(n)}$ and that cannot be
solved by general circuits of size $O(n)$ and depth $O(\log
n)$, which is a long-standing open problem in circuit complexity~\cite{Valiant77}. 
\Snote{added Valiant citation}

\subsection{Our Results}

In this paper, we construct the first generators with seed length
$\tO(\log(n/\eps))$ (where $\tO(\;)$ hides polylogarithmic factors in its argument) for several well-studied classes of
functions mentioned above.

\begin{itemize}
\item \prg s for combinatorial rectangles. Previously, it was known how to construct \hsg s with seed length $O(\log(n/\eps))$~\cite{LinialLSZ}, but the best seed length for \prg s was $O(\log n+\log^{3/2}(1/\eps))$~\cite{Lu}.

\item \prg s for read-once \cnf\ and \dnf\ formulas.
  Previously, De, Etesami, Trevisan, and Tulsiani~\cite{DeEtTrTu10}
  and Klivans, Lee and Wan~\cite{KlivansLW10} had constructed \prg s
  with seed length $O(\log n \cdot \log(1/\eps))$.

\item \hsg s for width 3 branching programs.
  Previously, Sima and Zak~\cite{SimaZa11} had constructed hitting set
  generators for width 3 branching programs with seed length $O(\log
  n)$ in case the error parameter $\eps$ is very large (greater than
  5/6).
\end{itemize}
As a corollary of our \prg\ for combinatorial rectangles we get improved hardness amplification in NP by combining our results with those of Lu Tsai and Wu \cite{LuTW07} - we refer to \sref{sec:cr} for details\footnote{We thank an anonymous referee for pointing out this application.}.\Mnote{Added remark about hardness amplification within in NP. The details are at the end of the combinatorial rectangles secton.}

\subsection{Techniques}\label{sec:tech}

Our generators are all based on a general new technique --- the
iterative application of ``mild'' (pseudo)random restrictions.

To motivate our technique, we first recall H\aa stad's switching lemma
\cite{Ajtai,FurstSaSi,Hastad}: if we randomly assign a $1-1/O(k)$
fraction of the variables of a $k$-\cnf , then the residual formula on
the $n/O(k)$ unassigned variables is likely to become a
constant. Ajtai and Wigderson \cite{AjtaiW} proposed the following natural approach to constructing \prg s for \cnf s:
construct a small pseudorandom family of restrictions that: 1) makes
any given \cnf\  collapse to a constant function with high probability;
and 2) ensures that the \cnf\  collapses to each constant function with
the {\em right}  probability as determined by the bias of the formula.
Known derandomizations of the switching lemma are far from optimal in
terms of the number of random bits needed
\cite{AjtaiW,AgrawalAIPR,GopalanMR}.
We will show that, for read-once \cnf s, such a pseudorandom restriction
can be generated using $\tO (\log (m/\epsilon))$ random bits.

We apply restrictions that only set a constant fraction of the
variables at a time. The novel insight in our construction is that
although we cannot set all the bits at one go from a small-bias
distribution, we can set a constant fraction of bits from such
a distribution and prove that the bias of the formula is
preserved (on average). Hence we use only $\tO(\log (m/\epsilon))$ truly
random bits per phase.  While such mild random restrictions do not drastically
simplify the formulas, we show that in each phase a suitable measure
of progress improves (e.g. most clauses will either be satisfied or will have reduced width), implying that the formula collapses to a
constant after $O(\log\log (m/\epsilon))$ steps; and so the total
randomness will be $\tO(\log (m/\epsilon))$. 
\Snote{added parenthetical comment on progress - please check that it is correct.  also, do we really collapse the formula to a constant or only to a polylogarithmic number of clauses?} 
The idea of setting a few
variables at a time is inspired by a recent \prg\  for hashing balls into
bins due to Celis, Reingold, Segev, and Wieder~\cite{CReingoldSW}.

We illustrate our technique below with a toy example.


\paragraph{A Toy Example.}
Consider a read-once \cnf\  formula $f$ of width $w$ with $m =
2^{w+1}$ clauses in which the variables appear in order (aka the
Tribes function of \cite{Ben-OrLi85}). \Snote{did I cite correct Ben-Or/Linial paper?}  That is,
\begin{align*}
f(x) = f_1(x_1,\ldots,x_w) \wedge f_2(x_{w+1},\ldots,x_{2w}) \wedge
\cdots \wedge f_m(x_{(m-1)w+1},\ldots,x_{mw})
\end{align*}
where each $f_i$ is the OR function. $f$ has constant bias and can be computed both
by a combinatorial rectangle and a
width-3 branching program. De \etal\ showed that fooling this function
with error $\eps$ using small-bias spaces requires seed-length
$\Omega(w\log(1/\eps)/\log\log(1/\eps))$.

Assume we partition the input bits into two parts: $x$ which contains
the first $w/2$ variables of each clause and $y$ which contains the
rest. Let $x \circ y$ denote the concatenation of the two strings. We
would like to show that for $\D$ a small-bias distribution and $\mc{U}$ the uniform distribution,
\begin{align}
\label{eq:fool-f}
\left|\E_{x \sim \D, y \sim \mc{U}}\left[f(x \circ y)\right] - \E_{x \sim \mc{U}, y
  \sim \mc{U}}\left[f(x \circ y)\right]\right| \leq \eps
\end{align}

A naive approach might be to view setting $y \sim \mc{U}$ as applying
a random restriction with probability $1/2$. If this simplified the
function $f$ to the extent that it can be fooled by small-bias spaces, we
would be done. Unfortunately, this is too much to hope for; it is not
hard to see that such a random restriction is very likely to give
another Tribes-like function with width $w/2$, which is not much easier to
fool using small bias than $f$ itself.

Rather, we need to shift our attention to the {\em bias function} of
$f$. For each partial assignment $x$, we define the
bias function $F(x)$ as
\begin{align}
\label{eq:def-bias-f}
F(x) = \E_{y \sim \mc{U}}[f(x \circ y)].
\end{align}
We can now rewrite Equation \eqref{eq:fool-f} as
\begin{align}
\label{eq:fool-bias-f}
\left|\E_{x \sim \D}\left[F(x)\right] - \E_{x \sim \mc{U}}\left[F(x)\right]\right| \leq \eps
\end{align}
Our key insight is that for restrictions as above, the function
$F$ is in fact easy to fool using a small-biased space. This is despite the fact that 
$F(x)$ is an average of functions $f(x \circ
y)$ (by Equation~(\ref{eq:def-bias-f})), most of which are Tribes-like and hence are not easy to fool.  \Snote{rephrased so as not to say ``we find this surprising'' and instead let the reader decide for herself}

Let us give some intuition for why this happens.
Since $f(x \circ y) = \prod_{i=1}^mf_i(x\circ y)$,
\begin{align*}
F(x) & = \E_{y \sim \mc{U}}[f(x \circ y)] =
\prod_{i=1}^m\E_{y \sim \mc{U}}[f_i(x\circ y)] = \prod_{i=1}^mF_i(x),
\end{align*}
where $F_i(x)$ is the {\em bias function} of the $i^{th}$
clause. But note that over a random choice of $y$, $f_i(x)$ is set to $1$
with probability $1 - 2^{-w/2}$ and is a clause of width $w/2$
otherwise. Hence
\begin{align*}
F_i(x) = \E_{y \sim \mc{U}}[f_i(x \circ y)] = 1 -
\frac{1}{2^{w/2}}  + \frac{\vee_{j=1}^{w/2}x_{w(i-1)+ j}}{2^{w/2}}.
\end{align*}
As a consequence, over a random choice of $x$, we now have
\begin{align*}
F_i(x) = \begin{cases}
1 & \ \text{w.p.}  \ 1 - 2^{-w/2}\\
1 - 2^{-w/2} & \ \text{w.p.} \ 2^{-w/2}
\end{cases}
\end{align*}
Thus each $F_i(x)$ is a random variable  with $\E_x[F_i(x)] = 1 -2^{-w}$
and $\Var_x[F_i(x)] \approx 2^{-3w/2}$. In contrast, when we
assign all the variables in the clauses at once, each $f_i(x)$
behaves like a Bernoulli random variable with bias $1 - 2^{-w}$. While
it also has $\E_x[f_i(x)] = 1 -2^{-w}$, the variance is much
larger: $\Var_x[f_i(x)] \approx 2^{-w}$. The qualitative difference
between $2^{-3w/2}$ and $2^{-w}$ is that in the former case, the sum
of the variances over all $2^{w+1}$ clauses is small ($2^{-w/2}$),
but in the latter it is more than $1$. We leverage the small total variance to
show that small-bias fools $F$, even though it does not fool $f$
itself. Indeed, setting any constant fraction $\alpha < 1$ of
variables in each clause would work.

We now sketch our proof that small-bias spaces fool $F$.
Let $g_i(x) = F_i(x) - (1 - 2^{-w}) $ be
$F_i$ shifted to have mean $0$, so that $\E_x[g_i(x)^2] =
\Var[F_i(x)]$. We can write
\begin{align}
\label{eq:f-sym}
F(x) & = \prod_{i=1}^m\left(1 - 2^{-w} + g_i(x)\right) =
\sum_{k=1}^mc_kS_k(g_1(x),\ldots,g_m(x))
\end{align}
where $S_k$ denotes the $k^{th}$ elementary symmetric
polynomial and $c_k \in [0,1]$.\footnote{In the toy example we are currently studying, an alternative and simpler approach is to write
$F_i(x) = (1-2^{-w/2})^{1-h_i(x)}$, where $h_i(x) = \vee_{j=1}^{w/2}x_{w(i-1)+ j}$ is the indicator for whether $x$ already satisfies the $i$'th clause on its own.  Then $F(x)=\prod_i F_i(x)$ expands as a power series in
$\sum_i (1-h_i(x)-2^{-w/2})$, and higher moment bounds can be used to
analyze what happens when we truncate this expansion.  However, this expansion
is rather specific to the highly symmetric Tribes function, whereas we are able to apply the expansion in terms of symmetric polynomials much more generally.}
\Snote{added footnote describing simpler expansion for Tribes function}

Under the uniform distribution, one can show that
\begin{align*}
\E_{x \sim \mc{U}}\left[\,\left|S_k(g_1(x),\ldots,g_m(x))\right|\,\right]   \leq
\left(\sum_{i=1}^m\E_{x\sim \mc{U}}\left[g_i(x)^2\right]\right)^{k/2}  \leq 2^{-wk/4}.
\end{align*}
Thus for $k \geq O( (\log n)/w)$, we expect each term in the summation
in Equation~(\ref{eq:f-sym}) to be $1/\poly(n)$. So we can truncate at
$d=O((\log n)/w)$ terms and retain a good approximation under the uniform distribution.  \Snote{introduced parameter $d$ for degree at which we truncate, and used it to clarify next paragraph}

Our analysis of the small-bias case is inspired by the {\em gradually
  increasing independence} paradigm of Celis \etal ~\cite{CReingoldSW},
developed in the context of hashing. Every monomial in the $g_i$'s of
degree at most $d$ depends on at most $wd =O(\log n)$
variables. A small-bias space provides an almost $O(\log n)$-wise independent distribution on the variables of $x$, so the $g_i(x)$'s will be almost $d$-wise
independent. This ensures that polynomials in $g_1(x),\ldots,g_m(x)$
of degree at most $d$ (such as $S_1,\ldots,S_d$) will behave like they do
under the uniform distribution. 
But we also need to argue that the $S_k$'s for $k>d$ have a small contribution to
$\E_{x \sim \D}\left[F(x)\right]$.  

Towards this end, we  prove the following inequality for any real numbers $z_1,\ldots,z_m$:
\begin{align*}
\text{If \ } |S_1(z_1,\ldots,z_m)| \leq \frac{\mu}{2} \text{ and }\ |S_2(z_1,\ldots,z_m)| \leq
\frac{\mu^2}{2},\ \text{then } |S_k(z_1,\ldots,z_m)| \leq \mu^k.
\end{align*}
The proof uses the Newton--Girard formulas (see \cite{CoxLO07}) which relate the symmetric
polynomials and power sums. This lets us repeat the same truncation
argument, provided that $S_1(g_1(x),\ldots,g_m(x))$ and
$S_2(g_1(x),\ldots,g_m(x))$ are tightly concentrated even under
small-bias distributions. We prove this concentration holds via suitable
higher moment inequalities.\footnote{These inequalities actually
  require higher moment bounds for the $g_i$'s. We ignore this issue in this
  description for clarity, and because we suspect that this
  requirement should not be necessary.} \Snote{softened ``believe'' to ``suspect'' in footnote}

This lets us show that small bias fools $F(x)$. By iterating this
argument $\log w$ times, we get a \prg\  for $f$ with polynomially small
error and seed-length $O((\log n)(\log w)) = O((\log n)(\log\log n))$.

\paragraph{Read-Once $\cnf$s.}
The case of general read-once \cnf s presents several additional challenges.
Since we no longer know how the variables are grouped into clauses, we
(pseudo)randomly choose a subset of variables to assign using
$\eps$-biased spaces, and argue that for most clauses, we will not
assign few variables. Clauses could now have very different sizes, and our
approximation argument relied on tuning the amount of independence (or
where we truncate) to the width of the clause. We handle this via an $\Xor$
lemma for $\eps$-biased spaces, which lets us break the formula into $O(\log\log
n)$ formulae, each having clauses of nearly equal size and
argue about them separately.

\paragraph{Combinatorial Rectangles.}
A combinatorial rectangle $f : [W]^m\rightarrow \zo$ is a function of
the form $f(x_1,\ldots,x_m) = \wedge_{i=1}^m f_i(x_i)$ for some
Boolean functions $f_1,\ldots,f_m$.  Thus, here we know which parts of
the input correspond to which clauses (like the toy example above),
but our clauses are arbitrary functions rather than ORs. To handle
this, we use a more powerful family of gradual restrictions. Rather than setting
$w/2$ bits of each co-ordinate, we instead (pseudo)randomly restrict the
domain of each $x_i$ to a set of size $W^{1/2}$.  More precisely, we
use a small-bias space to pseudorandomly choose hash functions
$h_1,\ldots,h_m : [W^{1/2}]\rightarrow [W]$ and replace $f$ with the
restricted function $f'(z_1,\ldots,z_m) = \wedge_{i=1}^m (f_i\circ h_i)(z_i)$.

\paragraph{Width $3$ Branching Programs.}
For width 3 branching programs, inspired by Sima and Zak~\cite{SimaZa11}
we reduce the task of constructing \hsg s for width 3
to that of constructing \hsg s for read-once \cnf\
formulas where we also allow some clauses to be parities.  Our
\prg\ construction for read-once \cnf s directly
extends to also handle such formulas with parities (intuitively
because small-bias spaces treat parities just like individual
variables).  The first step of our reduction actually works for any
width $d$, and shows how to reduce the the task of constructing
\hsg s for width $d$ to constructing hitting set
generators for width $d$ branching programs with sudden death, where
the states in the {\em bottom level} are all assumed to be Reject states.


\paragraph{Organization.}
\sref{sec:prelims} gives some preliminaries on pseudorandomness.
\sref{sec:sym} develops our main new technical tools for
constructing sandwiching approximators for symmetric functions. We
prove  an \Xor\ Lemma for $\eps$-biased spaces in  \sref{sec:xor}.

\sref{sec:cr} describes our \prg\ construction for combinatorial rectangles.
The reduction from hitting sets for width $3$ branching programs
to hitting sets for \cnf s with parity is in \sref{sec:prgbp}. The
generator for read-once \cnf s and for \cnf s with parity are presented in
\sref{sec:prgcnf} and \sref{sec:cnfx} respectively.

\section{Preliminaries}\label{sec:prelims}

We briefly review some notation and definitions. We use $x \sim \D$ to denote sampling $x$ from a distribution $\D$. For a set $S$, $x \sim S$ denotes sampling uniformly from $S$. By abuse of notation, for a function $G:\zo^s \rgta \zo^n$ we let $G$ denote the distribution over $\zo^n$ of $G(y)$ when $y \sim \zo^s$. For a function $f:\zo^n \rgta \reals$, we denote $\E[f] = \E_{x \sim \zo^n}[f(x)]$.


\paragraph{Hitting Set Generators and Pseudorandom Generators.}

\begin{Def}[Hitting Set Generators]
A generator $G: \zo^r \rightarrow \zo^n$ is an {\em $(\eps,\delta)$-hitting set generator} ($\hsg$)
for a class $\C$ of Boolean functions if for every $f \in \C$ such that $\E[f]
\geq \eps$, we have $\E_{x \sim G}f(x) \geq \delta$. We refer to $r$ as the seed-length of the generator and say $G$ is explicit if there is an efficient algorithm to compute $G$ that runs in time $\poly(n,1/\eps,1/\delta)$.
\end{Def}
Typically, our hitting set generators will be $(\eps,\delta)$ generators for some $\delta = \poly(\eps, 1/n)$. Given two functions $g,h : \zo^n \rightarrow \zo$ we say $g \leq h$ if $g(x) \leq h(x)$ for all $x \in \zo^n$. To prove that $G$ hits $h$, it suffices to show $G$ hits some function $g \leq h$.

\begin{Def}[Pseudorandom Generators]
A generator $G: \zo^r \rightarrow \zo^n$ is an {\em
  $\epsilon$-pseudorandom generator} ($\prg$) for a class $\C$ of
Boolean functions if for every $f \in \C$, $|\E[f] - \E_G[f(y)]| \leq
\eps$. We refer to $r$ as the seed-length of the generator and say $G$
is explicit if there is an efficient algorithm to compute $G$ that
runs in time $\poly(n,1/\eps)$. We say $G$ $\epsilon$-fools $\C$ and
refer to $\epsilon$ as the error.
\end{Def}

We shall make extensive use of {\it small-bias spaces}, introduced in
the seminal work of Naor and Naor \cite{NaorNa}. Usually these are
defined as distributions over $\zo^n$, but it is more convenient for
us to work with $\dpm^n$.
\begin{Def}
A distribution $\calD$ on $\dpm^n$ is said to be $\epsilon$-biased if
for every nonempty subset $I \subseteq [n]$, $|\E_{x \sim \dpm^n}[\,\prod_{i
    \in I} x_i\,]| \leq \epsilon$.
\end{Def}
There exist explicit constructions of $\epsilon$-biased spaces which
can be sampled from with $O(\log n + \log(1/\epsilon))$ random bits
\cite{NaorNa}. These give efficient pseudorandom generators for the
class of parity functions.

\begin{Def}
  Let $0 < \alpha, \delta < 1/2$. We say a distribution on $\calD$ on $2^{[n]}$ is $\delta$-almost independent with bias $\alpha$ if $I \lfta \calD$ satisfies the following conditions:
\begin{itemize}
\item For every $i \in [n]$, $\pr[i \in I] = \alpha$.
\item For any distinct indices $i_1,\ldots,i_k \in [n]$ and $b_1,\ldots,b_k \in \zo^k$,
\[ \pr\left[\,\wedge_{j=1}^k (\mathsf{1}(i_j \in I) = b_j) \,\right] = \prod_{j=1}^k \pr[\mathsf{1}(i_j \in I) = b_j] \pm \delta.\]
\end{itemize}
\end{Def}
There exist explicit constructions of distributions in $\calD$ as above which only need $O(\log n + \log(1/\alpha \delta))$ random bits \cite{NaorNa}. We will write $I \lfta \calD(\alpha,\delta)$ for short whenever $I$ is sampled from a $\delta$-almost independent distribution with bias $\alpha$ as above.

\paragraph{Sandwiching Approximators.} One of the central tools we use
is to construct {\em sandwiching polynomial approximations} for
various classes of functions. The approximating polynomials
$(P_\ell,P_u)$ we construct for a function $f$ will have two
properties: 1) low-complexity as measured by the ``$\Ll$-norm'' of
$P_\ell,P_u$ and 2) they ``sandwich'' $f$, $P_u \leq f \leq P_u$. The
first property will be important to argue that small-bias spaces {\em
  fool} the approximating polynomials and the second property will
allow us to lift this property to the function being approximated. We
formalize these notions below. For notational convenience, we shall
view functions and polynomials as defined over $\dpm^n$.

\begin{Def}
Let $P:\dpm^n \rgta \reals$ be a polynomial defined as $P(x) = \sum_{I \subseteq [n]} c_I \prod_{i \in I} x_i$. Then, the $\Ll$-norm of $P$ is defined by $\Ll[P] = \sum_{I \subseteq [n]} |c_I|$. We say $f: \pmo^n \rightarrow \R$ has {\em $\delta$-sandwiching approximations} of $\Ll$ norm $t$ if there exist functions $f_u,f_\ell:
\pmo^n \rightarrow \R$ such that
\eat{
\begin{align*}
f_\ell(x) \leq f(x)  \leq f_u(x) \ \forall x,\\
\E[f_u(x)] - \E[f_\ell(x)] \leq \delta\\
\Ll(f_\ell), \Ll(f_u) \leq t.
\end{align*}}
\begin{align*}
f_\ell(x) \leq f(x)  \leq f_u(x) \ \forall x,\ \ \ \E[f_u(x)] - \E[f_\ell(x)] \leq \delta,\ \ \ \Ll(f_\ell), \Ll(f_u) \leq t.
\end{align*}
We refer to $f_\ell$ and $f_u$ as the lower and upper sandwiching approximations to $f$ respectively.
\end{Def}

\eat{
\begin{Def}
Let $P:\dpm^n \rgta \reals$ be a polynomial defined as $P(x) = \sum_{I \subseteq [n]} c_I \prod_{i \in I} x_i$. Then, the $\Ll$-norm of $P$ is defined by $\Ll[P] = \sum_{I \subseteq [n]} |c_I|$.
\end{Def}
\begin{Def}
Let $f: \pmo^n \rightarrow \R$. We say that $f$ has {\em $\delta$-sandwiching approximations} of $\Ll$ norm $t$ if there exist functions $f_u,f_\ell:
\pmo^n \rightarrow \R$ such that
\eat{
\begin{align*}
f_\ell(x) \leq f(x)  \leq f_u(x) \ \forall x,\\
\E[f_u(x)] - \E[f_\ell(x)] \leq \delta\\
\Ll(f_\ell), \Ll(f_u) \leq t.
\end{align*}}
\begin{align*}
f_\ell(x) \leq f(x)  \leq f_u(x) \ \forall x,\ \ \ \E[f_u(x)] - \E[f_\ell(x)] \leq \delta,\ \ \ \Ll(f_\ell), \Ll(f_u) \leq t.
\end{align*}
We refer to $f_\ell$ and $f_u$ as the lower and upper sandwiching approximations to $f$ respectively.
\end{Def}}

It is easy to see that the existence of such approximations implies
that $f$ is $\delta + t\eps$ fooled by any $\eps$-biased
distribution. In fact, as was implicit in the work of Bazzi
\cite{Bazzi} and formalized in the work of De
et.~al.~\cite{DeEtTrTu10}, being fooled by small-bias spaces is
essentially equivalent to the existence of good sandwiching
approximators.

\begin{Lem}\cite{DeEtTrTu10}
\label{lem:lp}
Let $f:\pmo^n \rgta \R$ be a function. Then, the following hold for
every $0 < \epsilon < \delta$:
\begin{itemize}
\item If $f$ has $\delta$-sandwiching approximations of $\Ll$-norm at most $\delta/\epsilon$, then for every $\epsilon$-biased distribution $\calD$ on $\dpm^n$, $|\E_{x \sim \calD}[f(x)] - \E[f]| \leq \delta$.
\item If for every $\epsilon$-biased distribution $\calD$, $|\E_{x \sim \calD}[f(x)] - \E[f]| \leq \delta$, then, $f$ has $(2\delta)$-sandwiching approximations of $\Ll$-norm at most $|\E[f]| + \delta + (\delta/\epsilon)$\footnote{
De \etal\ actually show a bound of $\delta/\eps$ on the $\Ll$ norm
of the sandwiching approximators excluding their constant term. But it
is easy to see that the constant term of the approximators is bounded
by $|\E[f]| + \delta$.}.
\end{itemize}
\eat{
with $\E[f] = \E_\cU[f(x)]$.
Assume that the equation
\begin{align}
\label{eq:f-is-foolish}
|\E_\cU[f(x)] - \E_\cD[f(x)]| \leq \delta
\end{align}
holds for every $\eps$-biased distribution $\mc{D}$. Then $f$ has
$\delta$-sandwiching approximations of $\Ll$ norm $1/\eps +
|\E[f]| + \delta$.}
\end{Lem}

\paragraph{Pseudorandom Generators for \cnf s.}
A Conjunctive normal form formula (\cnf ) is a conjunction of disjunctions of literals. Throughout we view \cnf s as functions on $\dpm^n$, where we identify $-1$ with $\mathtt{false}$ and $1$ with $\mathtt{true}$. We say a \cnf\  $f = C_1 \wedge C_2 \wedge \cdots \wedge C_m$ is a read-once \cnf\  ($\rcnf$), if no variable appears (by itself or as is its negation) more than once. We call $m$ the size of $f$ and the maximum number of variables in $C_1,\ldots,C_m$ the width of $f$. We shall also use the following results of \cite{DeEtTrTu10}, \cite{KlivansLW10} which say that $\rcnf$s with small number of clauses have very good sandwiching approximators.

\begin{theorem}\label{th:smallcnf}
Let $f:\dpm^n \rgta \zo$ be a $\rcnf$ with at most $m$ clauses. Then,
for every $\epsilon > 0$, $f$ has $\epsilon$-sandwiching polynomials
with $\Ll$-norm at most $m^{O(\log(1/\epsilon))}$.
\end{theorem}

\begin{theorem}\label{thm:smallwcnf}
Let $f:\dpm^n \rgta \zo$ be a \cnf\  with at most $m$ clauses and width
at most $w$. Then, for every $\epsilon > 0$, $f$ has
$\epsilon$-sandwiching polynomials with $\Ll$-norm at most
$(m/\epsilon)^{O(w\log w)}$.
\end{theorem}

\section{Sandwiching Approximators for Symmetric Functions}
\label{sec:sym}

For $k \geq 1$, let $S_k:\R^m \rgta \R$ denote the $k^{th}$ elementary symmetric polynomial defined by 
\[ S_k(z_1,\ldots,z_m) = \sum_{I\subseteq [m], |I| = k}\, \prod_{i \in I} z_i.\]
Our main result on sandwiching approximators for symmetric functions
is the following:

\begin{Thm}
\label{thm:main}
Let $g_1,\ldots,g_m : \pmo^n \rgta \R$ be functions on disjoint sets of input variables and $\sigma_1,\sigma_2,\ldots,\sigma_m$ be positive numbers such that for all $i \in [m]$,
\begin{align*}
\E[g_i] = 0,\ \ \ \Ll[g_i] \leq t,\ \ \ \E_{x \sim \dpm^n}[(g_i)^{2k}] \leq (2k)^{2k}\sigma_i^{2k} \ \ \text{for } k \geq 1.
\end{align*}
Let $\sigma^2 = (\sum_i \sigma_i^2)/m$ and $\delta \in (0,1)$ and $\epsilon, k > 0$ be such that 
\begin{equation}
\label{eq:assumption}
m \sigma^2 \leq \frac{1}{\log(1/\delta)^{25}}, \ \ \ k = \left\lceil
\frac{5\log(1/\delta)}{\log(1/m\sigma^2)}\right\rceil, \ \ \ \eps = \frac{\delta^4}{(mt +1)^{2k}}. 
\end{equation}
Let $P(x) = \sum_{i=0}^mc_iS_i(g_1(x),\ldots,g_m(x))$ be a symmetric multilinear function of the $g_i$s that computes a bounded function $P:\pmo^n \rgta [-B,B]$, with $|c_i| \leq C$ for
all $i \in [m]$. \eat{Let 
\begin{align*}
k = \left\lceil
\frac{5\log(1/\delta)}{\log(1/m\sigma^2)}\right\rceil, \ \eps
= \frac{\delta^4}{(mt +1)^{2k}}. 
\end{align*}}
Then, 
\begin{enumerate}
\item For every $\eps$-biased distribution $\mc{D}$, we have
\begin{align*}
\left|\E_{x \sim \dpm^n}[P(x)] - \E_{x \sim \calD}[P(x)]\right| \leq O(B +C)\delta.
\end{align*}
\item $P$ has $O(B + C)\delta$ sandwiching approximations of $\Ll$
  norm $O((B+C)(mt +1)^{2k}\delta^{-3})$.
\end{enumerate}
\end{Thm}

As an illustration of this theorem, we state the following immediate
corollary which formalizes the argument for the
toy example in the introduction. 
 
\begin{Thm}\label{thm:corsym}
Let $\kappa > 0$ be a constant. Let $g_1,\ldots,g_m : \pmo^n \rgta [-\sigma,\sigma]$ be functions on disjoint sets of input variables with $\E[g_i] = 0$, $\Ll[g_i] = O(1)$ and $\sigma \leq 1/m^{-1/2-\kappa}$. Let $P:\dpm^n \rgta [-1,1]$ be a symmetric polynomial in $g_i$'s of the form $P(x) = \sum_{i=0}^m c_i S_i(g_1,\ldots,g_m)$,
with $|c_i| \leq 1$. Then, for every $\delta \in (0,1)$, with $\log
(1/\sigma) \geq \Omega_\kappa(\log(1/\delta))$, $P$ has
$\delta$-sandwiching polynomials of $\Ll$-norm at most
$\poly(1/\delta)$.  
\end{Thm}
To derive \tref{thm:corsym} from \tref{thm:main}, observe that in the
notation from \sref{sec:tech}, $m = 2^{w+1}$, $\sigma^2 \approx
2^{-3w/2}$ and all the other conditions hold. 

In the rest of this section, we prove the first statement of
\tref{thm:main}. The second statement follows from the first by
\lref{lem:lp}. We first sketch the steps involved in the proof.

Let $k, \epsilon$ be as in the theorem and let $\calD$ be a
$\epsilon$-biased distribution. Let $\Pk \equiv \sum_{i=0}^k c_i
S_i(g_1,\ldots,g_m)$. We will prove the theorem by showing that $P$
cannot distinguish the uniform distribution from $\calD$ by a series
of inequalities:  
\begin{align*}
 \E_{x \sim \dpm^n}\left[P(g_1(x),\ldots,g_m(x))\right]
 & \approx_{\delta} \E_{x \sim
   \dpm^n}\left[\Pk(g_1(x),\ldots,g_m(x))\right]\\ 
& \approx_{\delta} \E_{x \sim
  \calD}\left[\Pk(g_1(x),\ldots,g_m(x))\right]\\ 
& \approx_\delta \E_{x \sim \D}\left[P(g_1(x),\ldots,g_m(x))\right].
\end{align*}

Of these, the second inequality will follow from the fact that $\Ll[\Pk] = \poly(1/\delta)$ (this is not too hard). The first inequality can be seen as a special case of the last inequality as the uniform distribution is an $\epsilon$-biased distribution for any $\epsilon$. Much of our effort will be in showing the last inequality. 

To do this, we first show that there is an event $\mathcal{E}$ that
happens with high probability under any $\epsilon$-biased
distribution, and conditioned on which $\Pk$ is a very good
approximation for $P$. We then prove the last inequality by
conditioning on the event $\mathcal{E}$ and using Cauchy-Schwarz to
bound the error when $\mathcal{E}$ does not occur. The event
$\mathcal{E}$ will correspond to $|S_1(g_1,\ldots,g_m)|$,
$|S_2(g_1,\ldots,g_m)|$ being small, which we show happens with high
probability using classical moment bounds. Finally, we show that $\Pk$
approximates $P$ well if $\mathcal{E}$ happens by using the
Newton-Girard Identities for symmetric polynomials (see
\lref{lem:s1s2}).  

\subsection{Proof of \tref{thm:main}}

Our first task will be to show that under the assumptions of the theorem, $|\sum_i g_i(x)|$ and $|\sum_i g_i(x)^2|$ are small with high probability. We do so by first bounding the $k$'th moments of these variables and applying Markov's inequality. For this we will use Rosenthal's inequalities (\cite{Rosenthal}, \cite{JohnsonSZ85}, \cite{Pinelis94}) which state the following:  
\begin{Lem}
For independent random variables $Z_1,\ldots,Z_m$ such that $\E[Z_i]
=0$, and all $k \in \mathbb{N}$,
\begin{align}
\label{eq:mean0}
\E\left[\left(\sum_{i=1}^mZ_i\right)^{2k}\right] \leq (2k)^{2k} \max\left(\sum_{i=1}^m\E[Z_i^{2k}], \left(\sum_{i=1}^m\E[Z_i^2]\right)^k\right).
\end{align}
For independent non-negative random variables $Z_1,\ldots,Z_m$, and all $k \in \mathbb{N}$,
\begin{align}
\label{eq:non-neg}
\E\left[\left(\sum_{i=1}^mZ_i\right)^k\right] \leq k^k
\max\left(\sum_{i=1}^m\E[Z_i^k],
\left(\sum_{i=1}^m\E[Z_i]\right)^k\right).
\end{align}
\end{Lem}

\begin{Lem}
\label{lem:mom}
For all integers $k \geq 2$, 
\begin{align}
\E_{x \sim \dpm^n}\left[\left(\sum_{i=1}^mg_i(x)\right)^{2k}\right] \ \leq \ &
(2k)^{4k}\left(\sum_{i=1}^m\sigma_i^2\right)^k\label{eq:mom1}\\
\E_{\cU}\left[\left(\sum_{i=1}^m(g_i(x))^2\right)^k\right] \ \leq \ &
(2k)^{3k}\left(\sum_{i=1}^m\sigma_i^2\right)^k\label{eq:mom2}
\end{align}
\end{Lem}
\begin{proof}
Let $Z_i = g_i(x)$, $\cU$. Then, $Z_i$'s are independent mean-zero variables. Now, by Rosenthal's inequality, \eref{eq:mean0},
\begin{align*}
  \E\left[ \left(\sum_i Z_i\right)^{2k}\right] &\leq (2k)^{2k} \max\left(\, \sum_i \E[Z_i^{2k}], \left(\sum_i \E[Z_i^2]\right)^k\,\right)\\
&\leq (2k)^{2k} \max\left(\, \sum_i (2k)^{2k}\sigma_i^{2k}, \left(\sum_i 4\sigma_i^2\right)^k\,\right)\\
&\leq (2k)^{4k} \max\left(\, \sum_i \sigma_i^{2k}, \left(\sum_i \sigma_i^2\right)^k\,\right)\\
&= (2k)^{4k}\left(\sum_i \sigma_i^2\right)^k.
\end{align*}
\eat{
For the first bound, we apply Rosenthals' inequality for mean zero variables, \eref{eq:mean0}, with $Z_i =
g_i(x)$, $\cU$. Since $\E_\cU[g_i^{2k}] \leq (2k)^{2k}\sigma_i^{2k}$, we have
\begin{align*}
\sum_{i=1}^m\E_\cU[g_i^{2k}] \leq (2k)^{2k}\sum_{i=1}^m\sigma_i^{2k} \leq
(2k)^{2k}\left(\sum_{i=1}^m\sigma_i^2\right)^k.
\end{align*}
Plugging this into Equation \ref{eq:mean0}, we get
\begin{align*}
\E_\cU\left[\left(\sum_{i=1}^mg_i\right)^{2k}\right] \leq
(2k)^{4k}\left(\sum_{i=1}^m\sigma_i^2\right)^k. 
\end{align*}}
The second bound follows similarly by applying Rosenthal's inequality,\eref{eq:non-neg}, to the non-negative random variables $Z_i^2 = g_i^2$:
\begin{align*}
\E\left[ \left(\sum_i Z_i^2\right)^{k}\right] \leq k^k \max\left(\, \sum_i \E[Z_i^{2k}], \left(\sum_i \E[Z_i^2]\right)^k\,\right) \leq (2k)^{3k} \left(\sum_{i=1}^m\sigma_i^2\right)^k. 
\end{align*}
\eat{
To prove the second bound, we apply Rosenthal's inequality for non-negative random variables, \eref{eq:non-neg}, with $Z_i = g_i^2$. We have 
$$\E[Z_i] = \sigma_i^2, \ \E[Z_i^k] = \E_\cU[g_i^{2k}] \leq (2k)^{2k}\sigma_i^{2k}.$$
Hence we get
$$\sum_{i=1}^m\E[Z_i^k] = \sum_{i=1}^m\E_\cU[g_i^{2k}] \leq (2k)^{2k}\sum_{i=1}^m\sigma_i^{2k} \leq (2k)^{2k}(\sum_{i=1}^m \sigma_i^2)^k.$$ 
Plugging this into Equation \ref{eq:mean0}, we get
\begin{align*}
\E_\cU\left[\left(\sum_{i=1}^mg_i^2\right)^k\right] \leq
(2k)^{3k}(\sum_{i=1}^m \sigma_i^2)^k.
\end{align*}}
\end{proof}
\eat{
For ease of notation, we define the quantity
$$\sigma^2 = \frac{1}{m}\sum_{i=1}^m\sigma_i^2.$$}
A consequence of Lemma \ref{lem:mom} is the following:
\begin{Cor}
\label{cor:mom-eps}
For all $k \geq 2$, under any $\eps$-biased distribution $\mc{D}$, 
\begin{align}
\E_{\cD}\left[\left(\sum_{i=1}^mg_1\right)^{2k}\right] \ \leq \ &
(2k)^{4k}(m\sigma^2)^k + \eps(mt)^{2k}\label{eq:mom1-eps}\\
\E_{\cD}\left[\left(\sum_{i=1}^mg_1^2\right)^k\right] \ \leq \ &
(2k)^{3k}(m\sigma^2)^k + \eps(mt^2)^k\label{eq:mom2-eps}
\end{align}
\end{Cor}
\begin{proof}
Note that for any function $h:\dpm^n \rgta \R$, $\Ll[h^k] \leq (\Ll[h])^k$. Therefore, applying this inequality to $h\equiv \sum_i g_i$, we get $\Ll[\,(\sum_i g_i)^{2k}\,] \leq (mt)^{2k}$. The first inequality now follows from \lref{lem:mom} and \lref{lem:lp}. The second inequality follows similarly.
\end{proof}

Next we show that $|\sum_i g_i|$, $\sum_i g_i^2$ being small implies
the smallness in absolute value of $S_k(g_1,\ldots,g_m)$ for every $k
\geq 2$. Note that there is no probability involved in this
statement. 

\begin{Lem}
\label{lem:s1s2}
Let $z_1,\ldots,z_m$ be real numbers that satisfy
\begin{align*}
\left|\sum_{i=1}^m z_i\right| \leq \mu, \ \ \sum_{i=1}^mz_i^2 \leq \mu^2.
\end{align*}
Then for every $k \geq 2$ we have
\begin{align*}
|S_k(z_1,\ldots,z_m)| \leq \mu^k.
\end{align*}
\end{Lem}
\begin{proof}
To prove this lemma, we first bound the power sums
$E_k(z_1,\ldots,z_m)$ which are defined as 
$$E_k(z_1,\ldots,z_m) = \sum_{i=1}^mz_i^k.$$
Note that $E_1 = S_1$. We start by bounding $E_k$ for $k \geq 2$ using
the $\mathbb{L}_k$ norm inequalities 
\begin{align*}
|E_k(z_1,\ldots,z_m)|^\frac{1}{k} \leq \left(\sum_{i=1}^m|z_i|^k\right)^\frac{1}{k} \ \leq
\ \left(\sum_{i=1}^m z_i^2\right)^\frac{1}{2} = E_2(z_1,\ldots,z_m)^\frac{1}{2}
\end{align*}
Hence we have $|E_k(z_1,\ldots,z_m)| \leq \mu^k$.

The relation between the power sums and elementary symmetric
polynomials is given by the Newton-Girard identities (see \cite{CoxLO07}, Chapter 7.1 for instance) discovered in the 17th century.
\begin{equation}
\label{eq:newton} 
S_k(z_1,\ldots,z_m) =
\frac{1}{k}\sum_{i=1}^k(-1)^{i-1}S_{k-i}(z_1,\ldots,z_m)E_i(z_1,\ldots,z_m). 
\end{equation}

We use these to show by induction on $k$ that $|S_k| \leq \mu^k$. For
$k=2$, we have
\begin{align*}
S_2(z_1,\ldots,z_m) = \frac{1}{2}(S_1(z_1,\ldots,z_m)^2 - E_2(z_1,\ldots,z_m)) \leq \frac{1}{2}(\mu^2 + \mu^2) \leq \mu^2.
\end{align*}

Assume we have proved the bound up to $k-1$. Using the Newton-Girard formula,
\begin{align*}
|S_k(z_1,\ldots,z_m)| & \leq \frac{1}{k}\sum_{i=1}^k|S_{k-i}(z_1,\ldots,z_m)||E_i(z_1,\ldots,z_m)| \leq \frac{1}{k}\sum_{i=1}^k\mu^{k-i}\mu^i \leq \mu^k.
\end{align*}
\end{proof}

Let 
$$\Pk(x) = \sum_{i=0}^kc_iS_i(g_1,\ldots,g_m)$$
denote the truncation of $P$ to degree $k$. We use the following bounds for $\Pk$.
\begin{Lem}
\label{lem:pk}
Let $P,m,t,C,\D$ be as in \tref{thm:main}. If $m\sigma^2 \leq \frac{1}{2}$, then for every $k \in \mathbb{N}$,
\begin{align}
\E_\cD\left[\Pk(x)^2\right] \leq 2 C^2 + \eps\cdot(mt+1)^{2k}\cdot C^2.
\end{align}
\eat{
\E_\cU[\Pk(x)^2] \leq 2C^2,\\
\|\Pk\|_\infty & \leq  eCm^k
\end{align*}}
\end{Lem}
\begin{proof}
We observe that the symmetric polynomials $S_0 =1,\ldots, S_k$ on $g_1,\ldots,g_m$ are mutually orthogonal under the uniform distribution, i.e., for $i \neq j$,
$$\E_\cU[S_i(g_1(x),\ldots,g_m(x))\cdot S_j(g_1(x),\ldots,g_m(x))] = 0.$$ 
For brevity, we shall omit writing out the argument $x$ in the following. For $i \geq 1$, we have
$$\E_\cU\left[S_i(g_1,\ldots,g_m)^2\right] = \sum_{|S| =i}\prod_{j \in S}\E_\cU\left[g_j^2\right]
\leq \left(\sum_{j=1}^m \E_\cU[g_j^2]\right)^i \leq (m\sigma^2)^i.$$
Therefore, assuming that $m\sigma^2 \leq 1/2$,
\begin{align}
\label{eq:Esk1}
\E_\cU\left[\Pk(g_1,\ldots,g_m)^2\right] =
\sum_{i=0}^kc_i^2\E_\cU\left[S_i(g_1,\ldots,g_m)^2\right]  \leq
C^2\sum_{i=0}^k(m\sigma^2)^i \leq 2C^2.  
\end{align}

Since $\Ll[g_j] \leq t$, we have 
\begin{align*}
\Ll\left[S_i(g_1,\ldots,g_m)\right] & \leq {m \choose i}t^i,\\
\Ll\left[\Pk\right] & \leq C \sum_{i=0}^k{m \choose i}t^i \leq C\cdot (mt+1)^k, \text{ and}\\
\Ll\left[\Pk^2\right] & \leq \Ll\left[\Pk\right]^2 \leq C^2\cdot (mt+1)^{2k}.
\end{align*}

Hence
$$\E_\cD[\Pk(x)^2] \leq C^2(2+ \eps(mt+1)^{2k}) \leq  2 C^2 + \eps\cdot (mt+1)^{2k})\cdot C^2.$$

\eat{Since $|g_j| \leq 1$, we have 
\begin{align*}
|S_i(g_1,\ldots,g_m)| & \leq {m \choose i}\\
|\Pk(g_1,\ldots,g_m)| & \leq C\sum_{i=0}^k{m \choose i} \leq eCm^k.
\end{align*}}

\end{proof}

\paragraph{Setting Parameters.}
In Theorem \ref{thm:main}, we choose
\begin{align*}
k = \left\lceil \frac{5\log(1/\delta)}{\log(1/m\sigma^2)}\right\rceil
\end{align*}
which guarantees $\delta^5/2 \leq (m\sigma^2)^k \leq \delta^5$.
By \eref{eq:assumption} we have
\begin{align*}
m\sigma^2 \leq \frac{1}{\log(1/\delta)^{25}},
\end{align*}
from which it follows that
\begin{align}
\label{eq:param-5}
k & \leq  \frac{\log(1/\delta)}{5\log\log(1/\delta)}, \text{ and}\\
(2k)^{4k} & \leq \frac{1}{\delta}.
\end{align}
Finally, for all $\eps$ small enough so that $\eps\cdot(mt +1)^{2k} \leq \delta^4$, 
the following bounds will hold under the assumptions of \tref{thm:main}, by \cref{cor:mom-eps} and \lref{lem:pk},
\begin{align}
\label{eq:param-7:1}
\E_\cD\left[\Pk(x)^2\right] & \leq 4C^2,\\
\label{eq:param-7:2}
\E_{\cD}\left[\left(\sum_{i=1}^mg_i(x)\right)^{2k}\right] & \leq \
(2k)^{4k}(m\sigma^2)^k + \eps(mt)^{2k} \leq 2\delta^4\\
\label{eq:param-7:3}
\E_{\cD}\left[\left(\sum_{i=1}^mg_i(x)^2\right)^k\right] & \leq \
(2k)^{3k}(m\sigma^2)^k + \eps(mt)^{2k} \leq 2\delta^4.
\end{align}

We now proceed to prove Statement (1) in Theorem \ref{thm:main}, which we restate below with specific constants.

\begin{Lem}
\label{lem:1}
With the notation from \tref{thm:main}, we have
\begin{align}
\label{eq:p-is-foolish}
\left|\E_\cU[P(x)] - \E_\cD[P(x)]\right| \leq (4B + 13C)\cdot \delta.
\end{align}
\end{Lem}
\begin{proof}
We will show that under any $\eps$-biased distribution $\mc{D}$,
\begin{align}
\label{eq:p-approx}
\E_\cD[|P(x) - \Pk(x)|] \leq (2B +6C)\delta.
\end{align}
Note that $\mc{U}$ is $\eps$-biased for $\eps =0$, so the above bound applies to it. 
We derive \eref{eq:p-is-foolish} from \eref{eq:p-approx} as follows:
\begin{align}
\left|\E_\cU[P(x)] - \E_\cD[P(x)]\right| \leq \left|\E_\cU[P(x)] - \E_\cU[\Pk(x)]\right| +
\left|\E_\cU[\Pk(x)] - \E_\cD[\Pk(x)]\right| \nonumber\\
+ \left|\E_\cD[\Pk(x)] - \E_\cD[P(x)]\right|. \label{eq:triangle}  
\end{align}
The first and last terms are bounded using \eref{eq:p-approx}. We bound the middle term by
\begin{align*}
|\E_\cU[\Pk(x)] - \E_\cD[\Pk(x)]| \leq \eps \cdot \Ll[\Pk(x)] \leq \eps\cdot C\cdot (mt+1)^k \leq C\delta ^4.
\end{align*}
\eref{eq:p-is-foolish} follows by plugging these bounds into \eref{eq:triangle}: 
\begin{align*}
\left|\E_\cU[P(x)] - \E_\cD[P(x)]\right| \leq 2(2B + 6C)\delta + C\delta^4 \leq
(4B + 13C)\delta.
\end{align*}

We now prove \eref{eq:p-approx}. Define a good event $G
\subseteq \pmo^n$ containing those $x$ for which the following bounds hold: 
\begin{align}
\label{eq:good-event}
\left|\sum_{i=1}^m g_i(x)\right| \leq \delta^\frac{1}{k},
\ \left|\sum_{i=1}^m (g_i(x))^2\right|  \leq \delta^\frac{2}{k}.
\end{align}
For $x \in G$, $\Pk(x)$ gives a good approximation to
$P(x)$. By Lemma \ref{lem:s1s2}, we have $|S_\ell(g_1(x),\ldots,g_m(x))|
\leq \delta^{\ell/k}$ for all $\ell \geq 2$. Hence, for all $x \in G$
\begin{align}
\label{eq:good-bound}
|P(x) - \Pk(x)| \leq \sum_{\ell=k+1}^m|c_\ell S_\ell(g_1(x),\ldots,g_m(x))| 
\leq C\sum_{\ell = k+1}^{m}\delta^{\ell/k} \leq C\delta\sum_{\ell \geq 1}\delta^{\ell/k} \leq 2C\delta.
\end{align}

We now bound the probability of $\neg G$ using Markov's inequality applied to a $k$'th moment bound obtained from Equations \eqref{eq:param-7:2} and \eqref{eq:param-7:3}:
\begin{align*}
\Pr_\cD\left[\left|\sum_{i=1}^m g_i(x)\right| \geq \delta^{1/k}\right] =
\Pr_\cD\left[\left|\sum_{i=1}^m g_i(x)\right|^{2k} \geq \delta^2\right] \leq
\frac{1}{\delta^2}\E_\cD\left[\left(\sum_{i=1}^m g_i(x)\right)^{2k}\right]
\leq  2\delta^2,\nonumber\\
\Pr_\cD\left[\left|\sum_{i=1}^m g_i(x)^2\right| \geq \delta^{2/k}\right] =
\Pr_\cD\left[\left|\sum_{i=1}^m g_i(x)^2\right|^{k} \geq \delta^2\right] \leq
\frac{1}{\delta^2}\E_\cD\left[\left(\sum_{i=1}^m g_i(x)^2\right)^k\right]
\leq  2\delta^2,\nonumber\\
\end{align*}
Let $\ind{G}(x)$ and $\ind{\neg G}(x)$ denote the indicators of $G$
and $\neg G$ respectively. We have
\begin{align}
\label{eq:bad-bound}
\E_\cD[\ind{\neg G}(x)] \leq \Pr_\cD\left[\left|\sum_{i=1}^m g_i(x)\right| \geq
  \delta^{1/k}\right] + \Pr_\cD\left[\left|\sum_{i=1}^m g_i(x)^2\right| \geq
  \delta^{2/k}\right] \leq 4\delta^2.
\end{align}

Further,
\begin{align}
\label{eq:error-bound}
\E_\cD\left[\left|P(x) - \Pk(x)\right|\right] & = \E_\cD\left[\left|P(x) - \Pk(x)\right|\cdot\ind{G}(x)\right] + \E_\cD\left[\left|P(x) -
\Pk(x)\right|\cdot\ind{\neg G}(x)\right]
\end{align}
By Equation \ref{eq:good-bound}, we have 
\begin{align}
\label{eq:term1}
\E_\cD\left[\left|P(x) - \Pk(x)\right|\cdot \ind{G}(x)\right] \leq \max_{x \in G} |P(x) - \Pk(x)|
\leq 2C\delta 
\end{align}
To bound the second term, 
\begin{align}
\label{eq:term2}
\E_\cD\left[\left|P(x) - \Pk(x)\right|\cdot \ind{\neg G}\right] & \leq \E_\cD\left[\left|P(x)\right|\cdot\ind{\neg G}\right] +
\E_\cD\left[\left|\Pk(x)\right|\cdot\ind{\neg G}\right]\nonumber\\
& \leq \E_\cD[P(x)^2]^\frac{1}{2}\E_\cD[\ind{\neg G}]^\frac{1}{2} +
\E_\cD[\Pk(x)^2]^\frac{1}{2}\E_\cD[\ind{\neg G}]^\frac{1}{2}\nonumber\\
& \leq B\cdot 2\delta + 2C\cdot 2\delta
\end{align}
where we use the bounds
\begin{align*}
\E_\cD[P(x)^2]  \leq B^2 \ \ \ & (\text{Since } |P(x)| \leq B)\\
\E_\cD[\Pk(x)^2] \leq 4C^2 \ \ \ & (\text{\eref{eq:param-7:1}})\\
\E_\cD[\ind{\neg G}] \leq 4\delta^2 . \ \ \ & (\text{\eref{eq:bad-bound}})
\end{align*}
Plugging Equations \eqref{eq:term1} and \eqref{eq:term2} into \eref{eq:error-bound} we get \eref{eq:p-approx}.
\end{proof}

\section{An XOR Lemma for $\eps$-biased spaces}
\label{sec:xor}

In this section, we prove an \Xor\ Lemma that helps us show the existence of
good sandwiching approximators for the composition of a function on
few variables with functions on disjoint sets of variables, each of
which have good sandwiching approximators. We call it an \Xor\ lemma,
since one can view it as a generalization of Vazirani's \Xor\ lemma. 

\begin{Thm}
\label{thm:xor}
Let $f^1,\ldots,f^k :\pmo^n \rgta [0,1]$ be functions on disjoint input
variables such that each $f^i$ has $\eps$-sandwiching approximations
of $\Ll$ norm $t$.  Let $H:[0,1]^k \rgta [0,1]$ be a multilinear
function in its inputs. Let $h:\pmo^n \rgta [0,1]$ be defined as
$h(x) = H(f^1(x),\ldots,f^k(x))$. Then $h$ has $(16^k\eps)$-sandwiching
approximations of $\Ll$ norm $4^k(t+1)^k$.
\end{Thm}
\begin{proof}
For $S \subseteq [k]$ define the monomial 
$$M^S(x) = \prod_{i \in S}f^i(x)\prod_{j \not\in S}(1-f^j(x)).$$
Let $f^i_u$ and $f^i_\ell$ denote the upper and lower sandwiching
approximations to $f^i$. Then we have
\begin{align*}
f^i_u(x) \geq f^i(x),& \  \E_\cU[f^i_u(x) - f^i(x)] \leq \eps.\\ 
1 - f^j_\ell(x) \geq 1 - f^j_\ell(x),& \ \E_\cU[(1 - f^j_\ell(x)) - (1-
  f^j(x))] \leq \eps. 
\end{align*}
Hence, if we define 
\begin{align*}
M_u^S(x) = \prod_{i \in S}f^i_u(x)\prod_{j \not\in S}(1-f^j_\ell(x)), 
\end{align*} 
then we have
\begin{align*}
M_u^S(x) \geq M^S(x) \ \forall \ x \in \pmo^n,\\
\Ll[M_u^S] = \prod_{i \in S}\Ll[f^i_u]\prod_{j \not\in S}\Ll[1-f^j_\ell] \leq (t+1)^k.
\end{align*}
We will show using a hybrid argument, that
\begin{align*}
\E_\cU[M^S_u(x) - M^S(x)] \leq 2^k\eps.
\end{align*}
For simplicity, we only do the case $S = [k]$. We define a sequence of
polynomials $M^S_u = M_0,M_1 \ldots, M_k = M^S$ where 
$$M_i(x) = \prod_{j=1}^if^j(x)\prod_{j=i+1}^kf^j_u(x).$$
We now have
\begin{align*}
\E_\cU[M_i(x) - M_{i+1}(x)] & = \E_\cU \left[\left(f^{i+1}_u(x) -
  f^{i+1}(x)\right)\cdot \prod_{j=1}^if^j(x)\cdot \prod_{j=i+2}^kf^j_u(x)\right]\\ 
& = \E_\cU\left[f^{i+1}_u(x) - f^{i+1}(x)\right]\cdot \prod_{j=1}^i\E_\cU\left[f^j(x)\right]\cdot\prod_{j=i+2}^k\E_\cU\left[f^j_u(x)\right]\\
& \leq \eps\prod_{j=1}^i1 \prod_{j=i+2}^k(1 + \eps) \leq (1+ \eps)^{k-i-1}\eps
\end{align*}
where we use the facts that $\E_\cU[f^j] \leq 1$ and $\E_\cU[f^j_u] \leq
\E_\cU[f^j] + \eps \leq 1 +\eps$. We now have
\begin{align*}
\E_\cU\left[M^S_u(x) - M^S(x)\right] & \leq \sum_{i=0}^{k-1}\E_\cU[M_i(x) - M_{i+1}(x)] \\
& \leq \eps(1 + (1+\eps) \cdots (1+ \eps)^{k-1}) \leq 2^k\eps.
\end{align*}

To construct a lower-sandwiching approximator, we observe that
$$\sum_{S \subseteq [k]}M^S(x) = \prod_{i \in [k]}(f^i(x) + 1 -
f^i(x)) =1.$$
Hence if we define
$$M^S_\ell(x) = 1 - \sum_{T \neq S}M^T_u(x)$$
then 
\begin{align*}
M^S_\ell(x) & \leq 1 - \sum_{T \neq S}M^T(x) = M^S(x),\\
\E_\cU[M^S(x) - M^S_\ell(x)] & = \sum_{T \neq S}M^T_u(x) - M^T(x) \leq
4^k\eps,\\
\Ll[M^S_\ell] & \leq 2^k(t+1)^k.
\end{align*}

Finally, let $\ind{S} \in \zo^k$ denote the indicator vector of the set
$S$. Since $H$ is multilinear, we can write
\begin{align*}
H(y) = \sum_{S \subseteq [k]}H(\ind{S})\prod_{i \in S}y_i \prod_{j
  \not\in S}(1- y_j)
\end{align*}
where $H(\ind{S}) \in [0,1]$. Hence
\begin{align*}
h(x) = \sum_{S \subseteq [k]}H(\ind{S})\prod_{i \in S}f_i(x) \prod_{j
  \not\in S}(1- f_j(x)) = \sum_{S \subseteq [k]}H(\ind{S})M^S(x) 
\end{align*}
We define the polynomials
\begin{align*}
h_u(x) = \sum_{S \subseteq [k]}H(\ind{S})M^S_u(x), \ h_\ell(x) =
\sum_{S \subseteq [k]}H(\ind{S})M^S_\ell(x).
\end{align*}
It follows that 
\begin{align*}
h_u(x) & \geq h(x) \geq h_\ell(x)\\
\E_\cU[h_u(x) - h_\ell(x)] & \leq \sum_{S \subseteq
  [k]}H(\ind{S})\E_\cU[M^S_u(x) - M^S_\ell(x)] \leq 16^k\eps,\\
\Ll[h_u] & \leq 2^k(t+1)^k, \ \Ll[h_\ell] \leq 4^k(t+1)^k.
\end{align*}
\end{proof}


\newcommand{\bfs}{\ensuremath \bar{f}_s}
\newcommand{\brf}{\ensuremath \bar{f}}
\newcommand{\bx}{\ensuremath \bar{x}}
\newcommand{\samp}{\sim}
\newcommand{\crt}{\mathsf{CR}} 

\section{A \prg\  for Combinatorial Rectangles}\label{sec:cr}

We start by defining combinatorial rectangles ($\crt$s).
\begin{Def}
A {\em combinatorial rectangle} is a function $f:\left(\pmo^{w}\right)^m \rgta \zo$ of the form $f(x_1,\ldots,x_m) = \bigwedge_{i=1}^mf_i(x_i)$, where $f_i:\pmo^w \rgta \zo$,  and each $x_i \in \pmo^w$.
We refer to the $f_i$s as the {\em co-ordinate functions} of $f$. We refer to $m$ as the {\em size}\footnote{This is usually referred to as the {\em dimension} in the literature; we use this terminology for the \cnf\ analogy.} of $f$ and $w$ as the {\em width}.
\end{Def}
We construct an explicit \prg\  for $\crt$s with
seed-length $\tilde{O}(\log m + w + \log(1/\delta))$. The previous
best construction due to Lu had a seed-length of $O(\log m + w +
\log^{3/2}(1/\delta))$ \cite{Lu}. 
\begin{Thm}
\label{thm:prg-cr}
There is an explicit pseudorandom generator for the class of
combinatorial rectangles of width $w$ and size $m$ with error at most
$\delta$ and seed-length $O((\log w) (\log(m) + w + \log(1/\delta)) +
\log(1/\delta)\log\log(1/\delta)\log\log\log(1/\delta))$. 
\end{Thm}
\newcommand{\dpwm}{\left(\dpm^w\right)^m}
Our generator uses a recursive sampling technique and we next describe a single step of this recursive procedure. For this informal
description suppose that $\delta = 1/\poly(m)$, $w = O(\log m)$ and
let $v = 3w/4$. Fix a $\crt$ $f:\dpwm \rgta \zo$. 

Consider the following two-step process for generating a uniformly
element $x$ from $\dpwm$. 
\begin{itemize}
\item Choose a sequence of multi-sets $S_1,\ldots,S_m \subseteq \dpm^w$ each of size $2^v$ by picking $2^v$ elements of $\dpm^w$ independently and uniformly at random.
\item Sample $x_i \sim S_i$ and set $x = (x_1,\ldots,x_m)$. 
\end{itemize}
This results in an $x$ that is uniformly distributed over $\dpwm$. We
will show that the $\E_x[f(x)]$ will not change much, even if  the
sampling in the first step can is done pseudorandomly using a
small-bias space for suitably small $\eps$. 

\eat{
The core of our analysis and construction will be in showing that for
most choices of the sets $S_1,\ldots,S_m$, the randomness from the
second step is enough to {\em preserve the bias} of the $\crt$ $f$. We
will in fact show something far stronger: the same holds true even for
$S_1,\ldots,S_m$ sampled using $(1/\poly(m))$-almost independent
distributions (for the time being, we avoid technically describing
what {\em almost independent distributions} over multi-sets means). 
This sampling step will use roughly $O(\log m)$ random bits and it's
analysis will use \tref{thm:main}.}

Our final generator is obtained by iterating
the one-step procedure for $T = O(\log \log m))$ steps: At step $t$ we choose
multi-sets $S_1^t \subseteq S_1^{t-1},\ldots,S_m^t \subseteq
S_m^{t-1}$ each of cardinality exactly $2^{(3/4)^t w}$ using small-bias. 
After $T$ steps, we are left with a rectangle of width $w =O(\log\log
m)$. Such rectangles  can be fooled by
$\epsilon$-bias spaces where $\epsilon = 1/m^{O(\log\log m)}$. The
total randomness used over all the steps is $O((\log m)\cdot (\log \log m))$.

\subsection{Sandwiching Approximations for Bias Functions}

In the following, let $f$ be a $\crt$ of width $w$ and coordinate functions $f_1,\ldots,f_m:\dpm^w \rgta \zo$. We describe a restriction of $f$ which reduces the width from $w$ to $v = 3w/4$. 
\begin{itemize}
\item For every $a \in \pmo^v$, we sample string $x_a
  =(x_{a,1},\ldots,x_{a,m}) \sim  \{\pmo^w\}^m$.  
\item For $i \in [m]$, we define restricted co-ordinate functions
  $f^v_i$ on inputs $y_i$ by $f^v_i(y_i) = f(x_{y_i,i})$. 
\eat{
\begin{equation}
\label{eq:f-v-i}
f^v_i(y_i) = f(x_{y_i,i}).
\end{equation}}
\item Define  the restricted rectangle $f^v: \left(\pmo^{v}\right)^m \rgta \zo$
  on $y_1,\ldots,y_m$ by 
\begin{align}
\label{eq:f-v}
f^v(y_1,\ldots,y_m) = \bigwedge_{i=1}^mf^v_i(y_i)
\end{align}
\end{itemize}

Let $\bx \in \{\pmo^{w}\}^{2^v \times m}$ denote the matrix whose
rows are indexed by $a \in \pmo^v$,  the columns by $i\in
[m]$ and $(a,i)$'th entry is given by $\bx[a,i] = x_{a,i} \in
\dpm^w$. Every such matrix defines a {\em restriction} of $f$. We will
show that if choosing $\bx$ from an $\eps$-biased space for $\eps =
1/\poly(m)$ suitably small, and from the uniform distribution have
almost the same effect on $f$. For $i \in [m]$, let $\bx[i]$ denote
the $i$'th column of $\bx$. For each coordinate function $f_i$, define
the {\em sample average} function  
\begin{align}
\label{eq:sample-average}
\brf_i(\bx) = \frac{1}{2^v}\sum_{a \in \pmo^v}f_i(x_{a,i}) = \E_{a \sim \pmo^v}\left[f_i^v(a)\right].
\end{align}
\newcommand{\dpvm}{\left(\dpm^v\right)^m}
Note that each $\brf_i$ only depends on column $i$ of $\bx$. 
Define the {\em bias} function of $\bx$ as
\begin{align}
\label{eq:bias}
F(\bx) = \prod_{i=1}^m\brf_i(\bx) = \E_{y \sim \dpvm}\left[f^v(y)\right].
\end{align}

The main lemma of this section shows that this bias function can be fooled by
small-bias spaces. 
\begin{Lem}[Main]
\label{lem:main-cr}
Let $F$ be as defined in \eref{eq:bias}. Assume that $\delta < 1/4$ and $w \leq \log(1/\delta)$, $v = 3w/4 \geq 50 \log\log(1/\delta)$. 
\eat{
\begin{align*}
w \leq \log(1/\delta),\ v = \frac{3w}{4} \geq 50\log\log(1/\delta).
\end{align*}}
Then $F(x)$ has $\delta$-sandwiching approximations of $\Ll$
norm $\poly(1/\delta)$.
\end{Lem}

We start by stating two simple claims. 

\begin{Claim}\label{clm:l1brf}
For the sample average functions $\brf_i$ defined as in \eref{eq:sample-average}, we have 
\begin{align*}
\Ll(\brf_i) &\leq \Ll(f_i) \leq 2^{w/2}.\\
\E_\xu[\brf_i(\bx)] & = \E_{z \sim \dpm^w}[f_i(z)].
\end{align*}  
\end{Claim}
\begin{proof}
From \eref{eq:sample-average}, it follows that
\begin{align*}
\Ll[\brf_i] \leq \frac{1}{2^v}\sum_{a \in \pmo^v}\Ll[f_i] = \Ll[f_i]
\leq 2^{w/2}
\end{align*}
where the last inequality holds for any Boolean function on $w$ input
bits. The bound on the expectation follows directly from
\eref{eq:sample-average}. 
\end{proof}

The justification for the name bias function comes from the following
lemma.

\begin{Claim}
\label{clm:bias-of-f}
For $f^v$ and $F$ as defined in \eref{eq:f-v} and \eref{eq:bias}, 
\begin{align*}
\E_\yu[f^v(y)] = F(\bx).
\end{align*}
\end{Claim}
\begin{proof}
Note that 
\begin{align*}
f^v_i(y_i) = & \sum_{a \in \pmo^v}\ind{y =a}f_i(x_{a,i}),\\
\text{hence }  
\E_\yu[f^v_i(y)] = & \frac{1}{2^v} \sum_{a\in \pmo^v}f_i(x_{a,i}) = \brf_i(\bx).
\end{align*}
It follows that
\begin{align*}
\E_\yu[f^v(y)] & = \E_\yu\left[\bigwedge_{i=1}^mf^v_i(y_i)\right]\\
& = \E_\yu\left[\prod_{i=1}^mf^v_i(y)\right]\\
& = \prod_{i=1}^m\E_\yu[f^v_i(y)]\\
& =  \prod_{i=1}^m\brf_i(\bx)\\ 
& = F(\bx). 
\end{align*} 
\end{proof}

We will prove \lref{lem:main-cr} by applying \tref{thm:main} to the functions $g_i:\{\pmo^w\}^{2^v} \rgta \R$ defined as follows: $g_i(\bx) = (\brf_i(\bx) - p_i)/p_i,$, where $p_i = \E_{a \sim \dpm^w}[f_i(x)]$. (We assume $p_i \neq 0$.)
\eat{ Let
\begin{align*}
p_i & = \E_\cU[f_i(x)] \in [2^{-w},1-2^{-w}]\\
p  & = \E_\cU[f_i(x)] = \prod_{i=1}^mp_i
\end{align*}
\begin{align*}
g_i(\bx) = \frac{1}{p_i}(\brf_i(\bx) - p_i)
\end{align*}}

We will need the following technical lemma, which helps us show that
the functions $g_i$ satisfy the moment conditions needed to apply
\tref{thm:main}. For brevity, let $\mathcal{U}$ denote
$\left(\dpm^v\right)^{2^v \times m}$ in the remainder of this section.
\begin{Lem}
\label{lem:technical}
Let $p_i, g_i, \mathcal{U}$ be defined as above. We have 
\eat{
\begin{align*}
\Ll(\brf_i) &\leq \Ll(f_i) \leq 2^{w/2}.\\
\E_\xu[\brf_i(\bx)] & = \E_{a \sim \dpm^w}[f_i(a)] = p_i.
\end{align*}
Further we have }$\E_\xu[g_i(\bx)^{2k}] \leq
(2k)^{2k}\sigma_i^{2k}$ where
\begin{align*}
\sigma_i^2 = \begin{cases}
\frac{(1-p_i)}{2^vp_i} & \text{for } p_i \in    [2^{-v/10},1/2],\\
\frac{2(1-p_i)}{2^v} & \text{for } p_i \in    [1/2, 1 - 2^{-v}],\\
\frac{2}{2^{2v}} & \text{for } p_i \in    [1 - 2^{-v},1].
\end{cases}
\end{align*}
\end{Lem}
\begin{proof}
W start by
bounding the moments of $(\brf_i(\bx)- p_i)$. We have 
\begin{align*}
2^v(\brf_i(\bx) -p_i) = \sum_{a \in \pmo^v}(f(x_{a,i}) - p_i)
\end{align*}
which is the sum of $2^v$ i.i.d $p_i$-biased random variables with
mean $0$. Hence we can apply Rosenthal's inequality (\eref{eq:mean0}) to get
\begin{align*}
(2^v)^{2k}\E_\xu[(\brf_i(\bx) -p_i)^{2k}] \leq
  (2k)^{2k}\max\left(2^v((1-p_i)^{2k}p_i + p_i^{2k}(1 - p_i)), (2^v p_i(1-p_i))^k\right)
\end{align*}
Hence we have
\begin{align*}
\E_\xu[g_i(\bx)^{2k}] & = \frac{1}{p_i^{2k}}\E_\xu[(\brf_i(\bx)
  -p_i)^{2k}] \\
& \leq  \frac{(2k)^{2k}}{(2^v)^{2k}}\max\left(2^v\left(\left(\frac{1-p_i}{p_i}\right)^{2k}p_i + 1-p_i\right), \left(2^v\cdot\frac{1-p_i}{p_i}\right)^k\right)
\end{align*}
We will use the following bounds
\begin{align*}
2^v\left(\left(\frac{1-p_i}{p_i}\right)^{2k}p_i + 1-p_i\right) & \leq
\begin{cases} 
2^{v(1 + k/5)} & \text{for } p_i \in [2^{-v/10},1/2].\\
2^{v+1}(1 - p_i) & \text{for } p_i \in [1/2,1].
\end{cases}\\
\left(2^v\cdot \frac{1-p_i}{p_i}\right)^k & \leq
\left(2^{v+1}(1-p_i)\right)^k \ \ \text{for } p_i \in
     [1/2,1].
\end{align*}
From this it follows that $\E_\xu[g_i(\bx)^{2k}] \leq
(2k)^{2k}\sigma_i^{2k}$ where
\begin{align*}
\sigma_i^2 = \begin{cases}
\frac{(1-p_i)}{2^vp_i} & \text{for } p_i \in    [2^{-v/10},1/2],\\
\frac{2(1-p_i)}{2^v} & \text{for } p_i \in    [1/2, 1 - 2^{-v}],\\
\frac{2}{2^{2v}} & \text{for } p_i \in    [1 - 2^{-v},1].
\end{cases}
\end{align*}
\end{proof}

\begin{proof}[Proof of \lref{lem:main-cr}]
We first show the claim under the assumption that  $\E[f] = p \geq
\delta$ and later show how to get around this assumption.
Define the sets 
$$S_1 = \{i:p_i \in (0,2^{-v/10}]\},\ \ \ S_2 = \{i:p_i \in (2^{-v/10},1 -2^{-v}]\},\ \ \ S_3 = \{i:p_i \in (1-2^{-v},1]\}.$$
\eat{
\begin{align*}
S_1 & = \{i:p_i \in (0,2^{-v/10}]\},\\ 
S_2 &  = \{i:p_i \in (2^{-v/10},1 -2^{-v}],\\
S_3 & = \{i:p_i \in (1-2^{-v},1]\} 
\end{align*}}
For $j \in [3]$, let $F_j(\bx) = \prod_{i \in S_j}\brf_i(\bx)$ so that $F(\bx)
=\prod_{j=1}^3 F_j(\bx)$. 
\eat{
\begin{align*}
F_j(\bx) = \prod_{i \in S_j}\brf_i(\bx) \ \text{so that} \  F(\bx)
=\prod_{j=1}^3 F_j(\bx). 
\end{align*}}
We will construct sandwiching approximations for each $F_j$ and then
combine them via Theorem \ref{thm:xor}. We assume without loss of generality that $p_i \leq 1 - 2^{-w}$. Else, the $i$'th coordinate has bias $1$ and can be ignored without changing the rest of the proof.

\paragraph{Sandwiching $F_1$.}
We show that $\Ll[F_1]$ is itself small. Observe that $\delta \leq p = \prod_{i=1}^mp_i \leq \prod_{i \in S_1}p_i \leq 2^{-v|S_1|/10}$,
which implies that $|S_1| \leq 10 \log(1/\delta)/v$. Thus, by \clref{clm:l1brf},
$$\Ll[F_1] \leq \prod_{i \in S_1}\Ll[\brf_i] \leq 2^{\frac{w}{2}|S_1|}
\leq \left(\frac{1}{\delta}\right)^{5w/v} \leq \frac{1}{\delta^{20/3}}.$$

\paragraph{Sandwiching $F_2$.}
Note that $F_2(\bx) = \prod_{i \in S_2}\brf_i(\bx) = \prod_{i \in S_2}p_i\cdot (1 + g_i(\bx))$.
\eat{
\begin{align*}
F_2(\bx) = \prod_{i \in S_2}\brf_i(\bx) = \prod_{i \in S_2}p_i\cdot (1 + g_i(\bx)).
\end{align*}}
Notice that $F_2$ is a symmetric polynomial in the $g_i$'s, so we will obtain sandwiching polynomials for $F_2$ by applying \tref{thm:main} to $g_i$'s. As before, $\delta \leq p \leq \prod_{i \in S_2}p_i \leq (1 - 2^{-v})^{|S_2|}$, 
so we have $|S_2| \leq 2^v\log(1/\delta)$. Further we can write $\delta \leq p \leq \prod_{i \in S_2}(1 - (1- p_i)) \leq e^{- \sum_{i
    \in S_2} (1 - p_i)}$, so that $\sum_{i \in S_2}(1 - p_i)  \leq 2\log(1/\delta)$.
\eat{\begin{align*}
\delta \leq p \leq \prod_{i \in S_2}(1 - (1- p_i)) \leq e^{- \sum_{i
    \in S_2} (1 - p_i)},\\
\text{Hence} \ \ \sum_{i \in S_2}(1 - p_i)  \leq 2\log(1/\delta).
\end{align*}}
By Lemma \ref{lem:technical}, we have $\E_{\cU}[g_i(\bx)^{2k}] \leq (2k)^{2k} \sigma_i^{2k}$, where $2/2^v \leq \sigma_i^2 = (1-p_i)/2^{9v/10}$ for every $i \in S_2$. Hence,
\begin{align*}
\sum_{i \in S_2} \sigma_i^2 &= \sum_{i \in S_2}\frac{1
  -p_i}{2^{9v/10}} \leq \frac{2\log(1/\delta)}{2^{9v/10}} \leq
\frac{1}{\log(1/\delta)^{25}}.  
\end{align*} 
Hence \tref{thm:main} implies the existence of $O(\delta)$
 ($B = 1$ and $C = \prod_{i \in S_2} p_i \leq 1$) sandwiching approximations with $\Ll$ norm bounded by $(mt+1)^{2k}$ where
\begin{align*}
m = |S_2|  \leq 2^v\log(1/\delta) \leq 2^{5v/4}, \ \ \ t \leq 2^{w/2} \leq 2^v,\ \ \ k  \leq \frac{5\log(1/\delta)}{\log(1/\sum_i \sigma_i^2)} \leq \frac{25\log(1/\delta)}{4v}.  
\end{align*}
which implies the $\Ll$ norm is bounded by $\poly(1/\delta)$.

\paragraph{Sandwiching $F_3$.}
We write   
\begin{align*}
F_3(x) = \prod_{i \in S_3}\brf_i(x) = \prod_{i \in S_3}p_i(1 + g_i(x)).
\end{align*}
Note that each $i \in S_3$ satisfies $1-p_i \geq 2^{-w}$, which
implies that $|S_3| \leq 2^{w+1}\log(1/\delta))$. Let $\sigma_i^2 = 2/2^{2v} \geq \frac{1}{2^w}$. Then, by \lref{lem:technical}, $\E_{\xu}[g_i(x)^k] \leq (2k)^{2k} \sigma_i^{2k}$ and we have
\begin{align*}
\sum_{i \in S_3}\sigma_i^2 \leq \frac{2^{w+1}\log(1/\delta))}{2^{2v}} \leq
\frac{1}{2^{3v/5}} \leq \frac{1}{\log(1/\delta)^{25}}.
\end{align*}
By Theorem \ref{thm:main}, $F_3$ has $O(\delta)$ sandwiching
approximations with $\Ll$ norm bounded by $(mt+1)^{2k}$ where
\begin{align*}
m & \leq 2^w\log(1/p) \leq 2^{3v/2},\\ 
t & \leq 2^{w/2} \leq 2^v,\\ 
k & \leq \frac{5\log(1/\delta)}{\log(1/\sum_i \sigma_i^2)} \leq
\frac{25\log(1/\delta)}{3v}.  
\end{align*}
which implies the $\Ll$ norm is bounded by $\poly(1/\delta)$.

\paragraph{Sandwiching $F$.}
Since each $F_j$ has $O(\delta)$ sandwiching
approximations with $\Ll$ norm $\poly(1/\delta)$, by Theorem
\ref{thm:xor}, $F = F_1F_2F_3$ has $O(\delta)$
sandwiching approximations of $\Ll$ norm $\poly(1/\delta)$.

\paragraph{Handling all values of $\E[f]$.}
Finally, to get rid of the condition $\E[f] \geq \delta$, assume that
$\E[f] \leq \delta$. If $\E[f] = 0$, $f =0$ so there is nothing to
prove. If $\E[f] > 0$, every co-ordinate $f_i$ has at least one satisfying
assignment. We repeat the following procedure until the expectation
exceeds $\delta$: pick a co-ordinate $i$ which is not already the
constant $1$ function and add a new satisfying assignment to $i$. Such
a co-ordinate $i$ exists because $\delta < 1$. We
repeat this until  we get a rectangle $f^t$ such
that $\E[f^t]  \geq \delta$. Denote the resulting sequence 
$$f = f^0 \leq f^1 \cdots  \leq f^t.$$
We claim that for every $j$, 
$$\E_\xu[f^j(x)] \leq \E_\xu[f^{j+1}(x)] \leq 2\E_\xu[f^j(x)].$$ 
The last inequality holds since at each step, we at most double the
acceptance probability of the chosen co-ordinate, and hence of the
overall formula. Hence we have 
$$\E[f^t] \leq 2\E[f^{t-1}] \leq 2\delta.$$ 
We use the upper approximator for $f^t$ as the upper approximator for
$f$ and $0$ as the lower approximator. This gives sandwiching
approximators with error at most $2\delta$ and $\Ll$ norm $\poly(1/\delta)$.

This completes the proof of the lemma.
\end{proof}

\subsection{A Recursive Sampler for Combinatorial Rectangles}

We now use \lref{lem:main-cr} recursively to prove \tref{thm:prg-cr}. 
Our generator is based on a derandomized recursive sampling procedure
which we describe below. The inputs are the width $w$ and the size
$m$ of the rectangles we wish to fool and an error parameter $\delta
\leq 1/2^w$.  

\begin{enumerate}
\item Let $v_0 = w$, $v_j = \left(\frac{3}{4}\right)^jw$.  
\item While $v_j \geq 50\log\log(1/\delta)$ we sample $\bx_j \in \{\pmo^{v_{j-1}}\}^{2^{v_j} \times m}$ according to an $\eps_1$-biased
distribution for  $\eps \leq (1/\delta)^{c_1}$
for some large constant $c_1$. 
\item Assume that at step $t$ (where $t= O(\log w)$),  $v_t \leq
  50\log\log(1/\delta)$. Sample an input $\bx_t \in \left(\pmo^{v_{t-1}}\right)^m$ from
  an $\eps_2$-biased distribution where, for some large constant $c_2$,
$$\eps_2 \leq \left(1/\delta\right)^{c_2(\log\log(1/\delta)\log\log\log(1/\delta))}.$$   
\end{enumerate}
\newcommand{\bbx}{\mathbf{x}}

We next describe how we use $\bbx = (\bx_1,\ldots,\bx_t)$ to output an element of $\dpwm$. For $k \in \{1,\ldots,t-1\}$ we denote by $s_k$ the recursive sampling function which takes strings $\bx_j \in \{\pmo^{v_{j-1}}\}^{2^{v_j} \times m}$ for $j \in \{k+1,\ldots, t-1\}$  and $\bx_t \in \left(\pmo^{v_t}\right)^m$ and produces an
output string $s_k(\bx_{k+1},\ldots,\bx_t) \in \left(\pmo^{v_{k}}\right)^m$. Set $s_{t-1}(\bx_t) \equiv \bx_t$. Fix $k < t-1$ and let $z = s_{k+1}(\bx_{k+2},\ldots,\bx_t)$ be already defined. To define $s_k$, we will use $z$ to look up entries from the matrix $\bx_{k+1}$, so that the $i$'th coordinate of $s_k$ will be the entry of $\bx_{k+1}$ in the $z_i$'th row and $i$'th column:
\[ s_k(\bbx) \equiv s_k(\bx_{k+1},\ldots,\bx_t) =
\left((\bx_{k+1})_{z_1,1},
(\bx_{k+1})_{z_2,2},\ldots,(\bx_{k+1})_{z_m,m}\right) \in
\left(\dpm^{v_k}\right)^m.\] 

The above definition, though intuitive is a bit cumbersome to work
with. It will be far easier for analysis to fix the input
combinatorial rectangle $f:\dpwm \rgta \zo$ and study the effect of
the samplers $s_k$ on $f$. Let $f^0 =f$. Each matrix $\bx_j$ gives a
restriction of $f^{j-1}$: it defines restricted co-ordinate functions
$f_i^j: \pmo^{v_j} \rgta \zo$ and a corresponding restricted rectangle
$f^j:\{\pmo^{v_j}\}^m \rgta \zo$. We only use the following property
of the $s_j$s: 
\begin{align}
\label{eq:recursive}
f(s_0(\bbx)) = f^1(s_1(\bbx)) \cdots f^{t-1}(s_{t-1}(\bbx)).
\end{align}

To analyze the last step, we use the following corollary that follows
from \cite{DeEtTrTu10}. 

\begin{Cor}
\label{cor:gmr}
Every combinatorial rectangle $f:\{\pmo^v\}^m \rgta \pmo$ is
$\delta$-fooled by $\epsilon$-bias spaces for $\epsilon =
\left(m2^v/\delta\right)^{-O(v\log v)}$. 
\end{Cor}
\begin{proof}
Each co-ordinate function $f_i$ can be expressed as a $\cnf$ formula
with $2^v$ clauses of width $v$. Hence we can write $f$ as a $\cnf$ formula
with $m2^v$ clauses of width $v$. Now apply \tref{thm:smallwcnf}.
\end{proof}

For brevity, in the following let
\[ \mathcal{U} = \left(\dpm^{v_0}\right)^{2^{v_1} \times m} \, \times \,  \left(\dpm^{v_1}\right)^{2^{v_2} \times m}\,\times\, \cdots \,\times\, \left(\dpm^{v_{t-2}}\right)^{2^{v_{t-1}} \times m} \,\times\,\left(\dpm^{v_{t-1}}\right)^m,\]
be the domain of $\bbx$ as defined in the generator construction.

Let $\D^j$ denote the distribution on $\mathcal{U}$ where $\bx_i$ are
sampled from an $\eps$-biased distribution for $i < j$ and uniformly
for $i \geq j$. Then, $s_0(\D^0)$ is the uniform distribution on
$\{\pmo^w\}^m$ whereas $s_0(\D^{t})$ is the output of our Recursive
Sampler. 

\begin{Lem}
\label{lem:cr-correctness}
Let $f:\{\pmo^w\}^m \rgta \zo$ be a combinatorial rectangle with width
$w$ and size $m$. For distributions $\D^0$ and $\D^t$ defined
above, we have
\begin{align*}
\left|\E_{\bbx \samp \D^0}[f(s_0(\bbx))] - \E_{\bbx \samp \D^t}[f(s_0(\bbx))]\right| \leq \delta.
\end{align*}
\end{Lem}
\begin{proof}
Let $\delta' = \delta/t$. We will show by a hybrid argument that for all $j \in \{1,\ldots,t\}$ 
\begin{align}
\label{eq:hybrid}
\left|\E_{\bbx \samp \D^{j-1}}[f(s_0(\bbx))] - \E_{\bbx \samp \D^{j}}[f(s_0(\bbx))]\right| \leq \delta'.
\end{align}
In both $\D^{j-1}$ and $\D^{j}$, $\bx_i$ is drawn from an $\eps$-biased
distribution for $i< j$, and from the uniform distribution for $i >
j$. The only difference is $\bx_j$ which is sampled uniformly in $\D^{j-1}$ and from an
$\eps$-biased distribution in $\D^j$.

We couple the two distributions by drawing $\bx_i$ for $i < j$
according to an $\eps$-biased distribution. By Equation
\eqref{eq:recursive}, we get
\begin{align*}
\E_{\bbx\samp \D^{j-1}}[f(s_0(\bbx))]  \ = \E_{\bbx\samp \D^{j-1}}[f^{j-1}(s_{j-1}(\bbx))], \ 
\E_{\bbx\samp \D^j}[f(s_0(\bbx))] \ = \E_{\bbx\samp \D^j}[f^{j-1}(s_{j-1}(\bbx))]
\end{align*}
and our goal is now to show that
\begin{align}
\label{eq:hybrid-2}
\left|\E_{\bbx \samp \D^{j-1}}[f^{j-1}(s_{j-1}(\bbx))] - \E_{\bbx \samp
 \D^j}[f^{j-1}(s_{j-1}(\bbx))]\right| \leq \delta'. 
\end{align}

Define the bias function $F^{j-1}$ of the rectangle $f^{j-1}$ as in Equation
\eqref{eq:bias}. The string $\bx_j$ defines a restricted rectangle
$f^j:\{\pmo^{v_j}\}^m \rgta \zo$. Applying \clref{clm:bias-of-f}
we get
\begin{align*}
\E_{z \samp \left(\dpm^{v_j}\right)^m}[f^j(z)] = F^{j-1}(\bx_j).
\end{align*} 
In both distributions $\D^{j-1}$ and $\D^j$, $\bx_{j+1},\ldots,\bx_t$
are distributed uniformly at random, hence $s_j(\D^{j-1}) = s_j(\D^j) \sim \left(\dpm^{v_j}\right)^m$ are uniformly distributed, and this variable is independent of $\bx_j$. So
we have
\begin{align*}
\E_{x \samp \D^{j-1}}[f^{j-1}(s_{j-1}(x))] & = \E_{\bx_j \samp \D^{j-1}}\left[
  \E_{(\bx_{j+1}\ldots,\bx_t) \samp
    \D^{j-1}}[f^j(s_j(\bx_{j+1},\ldots,\bx_t))]\right] = \E_{\bx_j \samp
  \D^{j-1}}[F^{j-1}(\bx_j)],\\
\E_{x \samp \D^j}[f^{j-1}(s_{j-1}(x))] & = \E_{\bx_j \samp \D^j}\left[
  \E_{(\bx_{j+1}\ldots,\bx_t) \samp
    \D^j}[f^j(s_j(\bx_{j+1},\ldots,\bx_t))]\right] = \E_{\bx_j
  \samp \D^j }[F^{j-1}(\bx_j)]
\end{align*}
Thus it suffices to show that
\begin{align*}
\left|\E_{\bx_j \samp \D^{j-1}}[F^{j-1}(\bx_j)] - \E_{\bx_j \samp \D^j
}[F^{j-1}(\bx_j)]\right| \leq \delta'
\end{align*}
By Lemma \ref{lem:main-cr}, this holds true for $j \leq t-1$ provided
that $\eps_1 \leq \poly(1/\delta')$. 

For $j =t$, note that this is equivalent to showing that $\eps_2$-bias
fools the rectangle $f^t$. 
\eat{By Corollary \ref{cor:gmr} $f^t$ has $\delta$-sandwiching
approximations whose $\Ll$ norm is bounded by
\begin{align*}
\left(\frac{m2^{v_t}}{\delta}\right)^{O(v_t\log v_t)} \leq
  \left(\frac{1}{\delta}\right)^{O(\log\log(1/\delta)\log\log\log(1/\delta))}.
\end{align*}}
By \cref{cor:gmr}, $f^t$ is $\delta'$ fooled by $\eps_2$-biased spaces where
$$\eps_2 = \left(\frac{m2^{v_t}}{\delta'}\right)^{-O(v_t\log v_t)} = \left(\frac{1}{\delta'}\right)^{O(\log\log(1/\delta')\log\log\log(1/\delta'))}.$$  

Plugging these back into Equation \eqref{eq:hybrid}, the error is
bounded by $t\cdot \delta' \leq \delta$. 
\end{proof}

To complete the proof of Theorem \ref{thm:prg-cr}, we observe that the
total seed-length is
\begin{align*}
s & = O\left((\log w)(\log(m2^w/\eps_1)  + \log(m2^w/\eps_2)\right) \\
& = O\left(\,\log w \left(\log m + w + \log(1/\delta)\right) + \log(1/\delta)\log\log(1/\delta)\log\log\log(1/\delta)\right). 
\end{align*}  

We next state an application of our \prg\ to hardness amplification in $\mathsf{NP}$. Say that a Boolean function $f:\zo^n \rgta \zo$ is $(\epsilon,s)$-hard if any circuit of size $s$ cannot compute $f$ on more than a $1/2-\epsilon$ fraction of inputs. The hardness amplification problem then asks if we can use a mildly hard function in a black-box manner to construct a much harder function. Following the works of O'Donnell \cite{ODonnell04} and Healy, Vadhan and Viola \cite{HealyVV04}, Lu, Tsai and Wu \cite{LuTW07} showed how to construct $(2^{-\Omega(n^{2/3})}, 2^{-\Omega(n^{2/3})})$-hard functions in $\mathsf{NP}$ from $(1/\poly(n),2^{\Omega(n)})$-hard functions in $\mathsf{NP}$. Their improvement comes from using the \prg\ for combinatorial rectangles of Lu \cite{Lu} to partly derandomize the constructions of Healy, Vadhan and Viola. By using our \prg\ for combinatorial rectangles, \tref{thm:prg-cr}, instead of Lu's generator in the arguments of Lu, Tsai and Wu immediately leads to the following improved hardness amplification within NP.
\begin{Cor}
  If there is a balanced function in $\mathsf{NP}$ that is $(1/\poly(n), 2^{\Omega(n)})$-hard, then there exists a function in $\mathsf{NP}$ that is $(1/2^{n/\poly(\log n)}, 2^{n/\poly(\log n)})$-hard.
\end{Cor}
 

\section{\hsg s for Read-Once Branching Programs}\label{sec:prgbp}

In thsi section, we reduce the problem of constructing an \hsg\ for
width $3$ branching programs to the problem of \hsg\ construction for
\cnf\ formulas which are allowed to have parity functions as
clauses. We start with some definitions.   

A {\it read-once branching program} (ROBP) $B$ of {\it width} $d$ has
a vertex set $V$ partitioned into $n+1$ layers $V_0 \cup \cdots \cup
V_{n}$ where 
\begin{enumerate}
\item $V_0 = \{(0,0)\}$.
\item $V_t = \{(t,i)\}_{i \in [d]}$ for $t \in \{1,\ldots,n-1\}$.
\item $V_{n} = \{(n,1), (n,d)\}$.
\end{enumerate}
The vertex $(0,0)$ is referred to as the \Start\ state, while
$(n,1)$ and $(n,d)$ are referred to as \Acc\ and \Rej, respectively. 
Each vertex in $v \in V_t$ has two out-edges labeled $0$ and $1$,
which lead to vertices $N_0(v)$ and $N_1(v)$ respectively in $V_{t+1}$. 
We refer to the set of states $\{(t,1)\}_{t=1}^n$ as the {\it top level} and
$\{(t,d)\}_{t=1}^n$ as the {\it bottom level}.

A string $x \in \zo^n$ defines a path in $V_0 \times \cdots \times
V_{n}$ beginning at \Start\ and following the edge labeled $x_i$
from $V_i$. Let $\Path(x) = \Path_0(x),\ldots,\Path_{n}(x)$ denote this sequence of states, i.e., $\Path_1(x) = (0,0)$, and $\Path_{i+1}(x) = N_{x_i}(\Path_i(x))$. The string $x$ is {\it accepted} if $\Path_{n}(x) = \Acc$. Thus the branching program naturally computes a function $f:\zo^n \rightarrow \zo$. Let $\E[f] =\E_{x \sim \zo^n}[f(x)] = Pr_x[f(x) =1]$.  

Let $\bp{d,n}$ denote the set of all $f:\zo^n \rightarrow \zo$ that can
be computed by width $d$ ROBPs. Our hitting set generator for
$\bp{3,n}$ uses a reduction to the problem of hitting $\mathsf{CNF}$
formulas where clauses can be disjunctions of variables or parity
functions.  

\begin{Def}
Let $\cnfx(n)$ denote the class of read once formulas $f:\zo^n \rgta
\zo$ of the form $f  = \wedge_{i=1}^mT_i$ where each $T_i$ is either a
disjunction of literals or a parity function of literals and the
$T_i$s are on disjoint variables.
\end{Def}

\eat{
\begin{Thm}
For every $f \in \bp{3,n}$ where $\E[f] \geq \eps$, 
there exists $g \in \cnfx(n)$ such that $g \leq
f$ and $\E[g] \geq \poly(\eps,1/n)$.
\end{Thm}}

\begin{Thm}\label{th:bpmain}
For every $f \in \bp{3,n}$ there is an integer $k$ and $g\in
\cnfx(n-k)$ such that $0^k \circ g^{-1}(1) \subseteq f^{-1}(1)$ and
$\E[g] \geq (\E[f]/n)^{O(1)}$. 
\end{Thm}

Given this reduction, we get a \hsg\  for $\bp{3,n}$ by using the
\prg\ for $\cnfx$ that we construct in \tref{thm:prgcnfx}:
\ 
\begin{Thm}\label{th:bphsg}
 For every $\epsilon > 0$, there exists an explicit $(\epsilon,(\epsilon/n)^{O(1)})$-\hsg\ $G:\zo^r \rgta \zo^n$ for $\bp{3,n}$ with a seed-length of $O((\log (n/\epsilon)) \cdot (\log\log(n/\epsilon))^3)$.
\end{Thm}
We remark that using similar techniques, we can also achieve a seed-length of $O((\log n)(\log (1/\epsilon)))$ which is better than the above bound for large values of $\epsilon$. We defer the details of this to the full version.

The reduction in \tref{th:bpmain} is carried out in three steps. 
\begin{itemize} 
\item The first step (for the sake of \hsg s) reduces arbitrary width $3$ programs to ``sudden death'' width $3$ programs, where the last state in every layer is a \Rej\ state. (This step in fact works for all widths.)
\item The second step reduces ``sudden death'' width $3$ programs to intersections of width $2$ programs.
\item The third step reduces intersections of width $2$ programs to $\cnfx$ formulae.
\end{itemize}

\subsection{Reduction to Branching Programs with Sudden Death}
\begin{Def}
  A width $d$ BP with {\em sudden death} is a BP where the bottom level states are all $\Rej$ states. Formally this means $N_0((t,d)) = N_1((t,d)) = (t+1,d)$ for all $t = 1,\ldots,n-1$. Let $\bpr{d,n}$ denote the set of functions computable by such programs. 
\end{Def}
We reduce the problem of constructing hitting sets for width $d$ BPs to for ones with sudden death.

\eat{
\begin{Thm}\label{th:bptosd}
Let $G':\zo^s \rightarrow \zo^n$ be a $(\eps^2/2n,\delta)$-hitting set generator for $\bpr{d,n}$. Define, $G:\zo^{\log n + s} \rightarrow \zo^n$ as follows:
\begin{itemize}
\item Sample $k \sim [n]$ and $y \sim \zo^s$. 
\item Output $k$ $0$s followed by the first $n-k$ bits of $G'(y)$. 
\end{itemize}
Then, $G$ is a $(\eps,\delta/n)$-hitting set generator for $\bp{d,n}$.   
\end{Thm}}

\begin{Thm}\label{th:bptosd}
For every $f \in \bp{d,n}$ there is an integer $k$ and a $g:\zo^{n-k} \rgta \zo$, $g \in \bpr{d,n}$ such that $0^k \circ g^{-1}(1) \subseteq f^{-1}(1)$ and $\E[g] \geq \E[f]^2/2n$.
\end{Thm}

We first setup some notation. For a vertex $v \in V$ let $p(v)$ denote the probability of reaching \Acc\ starting from $v$ over a uniformly random choice of $x_{i+1},\ldots,x_n$. We call a state $v \in V$ such that $p(v) = 0$ a \Rej\ state. 
We order states in $V_t$ so that 
$$p((t,1)) \geq p((t,2)) \cdots \geq p((t,d)).$$ 
By definition, 
$$p(v) = \frac{1}{2}(p(N_0(v)) + p(N_1(v))).$$
It follows that
$$\E[f] = p((0,0)) \leq p((1,1)) \leq \cdots \leq p((n,1)) = 1,$$
$$\E[f] \geq p((1,d)) \geq p((2,d)) \geq \cdots \geq p((n,d)) = 0$$

Observe that, if $v \in V_j$ is such that $p(v) \leq \mu$, then $p((i,d)) \leq \mu$ for all $i \geq j$. 
\eat{
\begin{Def}
A width $d-1$ branching program with sudden reject 
is a width $d$ branching program such that $p((t,d)) = 0$ for every $t \in \{2,\ldots,n\}$. 
Let $\bpr{d-1}$ denote the set of functions computable by
such programs. 
\end{Def}}

\begin{Lem}
\label{lem:small-reject}
Let $B \in \bp{d,n}$. Let $R$ be a set of states such that $p(v) \leq
\mu \ \forall v\in R$ and let $j$ be the first layer such that $R \cap V_j \neq \emptyset$. Let $B'$ be obtained from $B$ by converting all states in $R$ into \Rej\ states by redirecting the edges out of $v \in R \cap V_i$, $i \geq j$, to $((i+1,d))$. Let $p'(v)$ denote the accepting probabilities of vertices in
$B'$. Then for all $v \in V$, we have $p'(v) \geq p(v) - \mu$.
\end{Lem}
\begin{proof}
If $p(v) \leq \mu$ the claim is trivial, so fix $v$ such that $p(v) >
\mu$. Let $R(x)$ denote the event that we visit a vertex in $R$ if we
follow $x$ from $v$ in $B$ and let $u(x)$ denote the first vertex in $R$ that is visited by this path. Let $\Acc(x)$ denote the event that $B$ accepts. 
We have
\begin{align*}
\Pr_x[R(x) \wedge \Acc(x)] & = \sum_{r \in R}\Pr_x[u(x) = r \wedge \Acc(x)]\\
&= \sum_{r \in R}\Pr_x[u(x) = r] \cdot \Pr_x[\Acc(x)|u(x) =r] \leq \sum_{r \in R}\Pr_x[u(x) = r] \cdot \mu \leq \mu,
\end{align*}
where we  use $\Pr_x[\Acc(x)|u(x) = r] = p(r) \leq \mu$ for all $r \in R$. 
But then 
\begin{align*}
\Pr_x[\Acc(x) \wedge \overline{R(x)}] = \Pr_x[\Acc(x)] -  \Pr_x[\Acc(x)\wedge
  R(x)] \geq p(v) - \mu.
\end{align*}
Finally, note that if we accept $x$ without ever reaching $R$ in $B$, then $x$ is also accepted by
$B'$. Hence $p'(v) \geq p(v) - \mu$.
\end{proof}

\begin{proof}[Proof of \tref{th:bptosd}]
Let $B$ be a branching program computing a function $f$ so that
$\E[f] \geq \eps$. Let $i$ denote the first layer where $p((i,d)) \leq \eps/2$. Note that $i \leq n$ since $p((n,d)) = 0$.
Every state $v$ up to layer $i-1$ satisfies
$p(v) \geq \eps/2$. Further, for every $j \geq i$, $p((i+1,d))\leq \eps/2$. Fix $k = i-2$ and let $v$ be the
state in level $i-1$ reached from \Start\ on the string $0^k$. 
Consider the branching program $B'$ of length $n' = n-k$ where we make
$v$ the new start state and keep the rest of the program
unchanged. The vertex set of $B'$ is $V' = \{v\}\cup_{j=i}^{n+1}V_j$ and 
it computes $f':\zo^{n'}\rightarrow \zo$ such that 
$$\E_{y \in \zo^{n'}}[f'(y)] = p(v) \geq \eps/2.$$

Thus, a random walk starting at $v$ reaches the top level
with probability at least $\eps/2$ (since this is a necessary
condition for $B'$ to accept). For $j \in \{i,\ldots,n-1\}$, let $q(j)$
denote the probability that we reach the top level for the first time
at layer $j$. So
$$\sum_{j=i}^{n+1}q(j) \geq \eps/2.$$
Hence there exists $j$ so that $q(j) \geq \eps/2n$.

We now make the following modifications to $B'$ to get a program $B''$
which is a width $d$ program with sudden death:
\begin{itemize}
\item For $t \in \{i,\ldots,j-1\}$ we convert the states $(t,1)$ into
  \Rej\ states. 
\item For $t \in \{j,\ldots,n+1\}$ we convert the states $(t,d)$ into
  \Rej\ states. 
\end{itemize}
We don't need to add an additional layer for making these modifications since we are turning one state in each layer to a \Rej\ state. 

It is clear that $B''$ computes a function $f''\leq f'$. 
Our goal is to show that $B''$ accepts a large
subset of inputs accepted by $B'$. Indeed, we claim that 
$$\E_{y \in \zo^{n'}}[f''(y)] \geq \frac{\eps^2}{4n}.$$

We observe that the probability that a random walk starting at $v$
reaches the top level for the first time in layer $j$ is the same in
$B''$ as in $B'$, hence it equals $q(j) \geq \eps/2n$. Further, using Lemma \ref{lem:small-reject} (to the sub-program of $B'$ starting at $(j,1)$) we claim
that 
$$p''(j,1) \geq p'(j,1) - \eps/2 \geq \eps/2$$
where we use the fact that $p'(j,1) = p(j,1) \geq p(1,1) \geq \eps$.
Note that the probability that $B''$ accepts is at least $q(j)
p''(j,1) \geq \eps^2/4n$, which comes from strings which
reach state $(j,1)$ and then reach \Acc.

The theorem now follows by setting $g \equiv f''$. By definition, $f'' \in \bpr{d,n-k}$ and 
$$\mathsf{0}^k \circ (f'')^{-1}(1) \subseteq \mathsf{0}^k \circ (f')^{-1}(1) \subseteq f^{-1}(1).$$ 

\eat{Assume that we guess $k$ correctly, which happens with probability
$1/n$. $G$ then simulates $B'$ on the string $G'(y)$.
Since $G'$ is an $(\eps^2/2n,\delta)$ generator for
$\bpr{d,n}$ and $f'' \in \bpr{d,n}$, we have 
$$\Pr_{y \in \zo^s}[f'(G'(y)) =1] \geq \Pr_{y \in \zo^s}[f''(G'(y)) =1]
\geq \delta.$$
Thus
$$\Pr_{k,y \in \zo^s}[f(G(y)) =1] \geq \frac{\delta}{n}.$$}
\end{proof}

\subsection{From $\bpr{3}$ to Intersections of $\bp{2}$}
We now reduce width $3$ programs with sudden death to intersections of width $2$ programs. 

\begin{theorem}\label{th:bpsdtobpint}
 Let $f:\zo^n \rgta \zo$ be in $\bpr{3,n}$. Then, there exists a function $g:\zo^n \rgta \zo$ that is an intersection of functions in $\bp{2,n}$ such that $g \leq f$ and if $p = \E[f]$, then $\E[g] \geq (p/2)^{13}$. 
\end{theorem}

Throughout this section, we are given $B \in \bpr{d,n}$ computing $f:
\zo^n \rightarrow \zo$. Let $\Bad$ denote the set of non-reject states that have
an out-edge leading to a $\Rej$ state (which are all states such that $p(v)
=0$). Further for each $x \in \zo^n$, let $\Bad(x)$ denote the number
of $\Bad$ states visited by $x$.

\begin{Lem}
We have
$$\Pr_{x \sim \zo^n}[\Bad(x) \geq t] \leq 2^{-t+1}.$$
\end{Lem}
\begin{proof}
Suppose that $t \geq 1$. For $i \in [n]$, let $Y_i$ denote the number of vertices in $\Bad$ visited by $\Path(x)$ in the first $i$ layers. Then, $Y_n = \Bad(x)$. We claim that,
\begin{equation}
  \label{eq:badprob}
 \pr[ Y_i = Y_{i+1} = Y_{i+2} = \cdots = Y_n \,|\,Y_i, \Path_i(x) \in \Bad] \geq 1/2.  
\end{equation}

This is because, if $\Path_i(x) \in \Bad$, then with probability at least $1/2$, $\Path_{i+1}(x)$ is a \Rej\ state, in which case $\Path_j(x)$ is a \Rej\ state for every $j \geq i+1$.

Further, if $Y_n \geq t$, then there must be an index $i < n$, where $Y_i \geq t-1$, $\Path_i(x) \in \Bad$ and $Y_n > Y_i$ (for instance $i$ can be the least $j$ such that $Y_j = t-1$). Therefore,
\begin{align*}
  \pr[Y_n \geq t] &= \pr[\,(\exists i < n,\, Y_i \geq t - 1,\, \Path_i(x) \in \Bad)\, \wedge \, (Y_n > Y_i)]\\
&= \pr[\, (\exists i < n, \, Y_i \geq t-1,\, \Path_i(x) \in \Bad)\,] \cdot \pr[ Y_n > Y_i\, | \,Y_i, \Path_i(x) \in \Bad]\\
&\leq \frac{1}{2} \cdot \pr[\, (\exists i < n, \, Y_i \geq t-1,\, \Path_i(x) \in \Bad)\,] \leq \frac{1}{2} \cdot \pr[Y_n \geq t-1],
\end{align*}
where the last two inequalities follow from \eref{eq:badprob} and the fact that $Y_i$'s are non-decreasing.
The claim now follows by induction.
\eat{
Let $v_1,\ldots,v_t$ denote the first $t$ $\Bad$ states visited by
$x$, occurring in layers $n_1 < \cdots < n_t$. Note that $v_i$ and
$n_i$ are themselves random variables which depend on
$x_1,\ldots,x_{n_i -1}$ but not on $x_{n_i}$ which is yet to be read.
We have
$$\Pr_{x \in \zo^n}[\Bad(x) \geq t] = \Pr_{x \in \zo^n}[x \text{
    visits }v_1,\ldots,v_t ] = \prod_{i=1}^{t}\Pr_{x \in
  \zo^n}[x \text{ visits }v_i| x \text{ visits } v_1,\ldots,v_{i-1}].$$

Note that $x_{n_{i-1}}$ is independent of $v_1,\ldots,v_{i-1}$. There
is (at least) one setting of $x_{n_{i-1}}$ which causes $x$ to reach
$\Rej$ from $v_{i-1}$ without visiting $v_i$. Hence for $i \geq 2$,
$$\Pr_{x \in \zo^n}[x \text{ visits }v_i| x \text{ visits }
  v_1,\ldots,v_{i-1}] \leq \frac{1}{2}$$
which gives the desired claim.}
\end{proof}

\begin{Cor}
\label{cor:exp-bad}
Let $\Pr_{x \sim \zo^n}[f(x) = 1] = p$. Then
$$\E_{x \sim f^{-1}(1)}[\Bad(x)] = \E_{x \sim \zo^n}[\Bad(x)|f(x) =1] \leq 2\log(2/p).$$
\end{Cor}
\begin{proof}
We have
$$\Pr_{x}[\Bad(x) \geq t|f(x) = 1] = \frac{\Pr_{x}[(\Bad(x) \geq t)
    \text{ and }(f(x) =1)]}{\Pr_x[f(x) =1 ]} \leq \frac{1}{2^{t-1}p}.$$
Let $t^* = \log(2/p)$. We then bound
\begin{align*}
\E_{x}[\Bad(x)|f(x) =1] &= \sum_{t \geq 0} t \cdot \Pr_{x}[\Bad(x) = t|f(x)=1] \\
&\leq t^* + \sum_{t > t^*} t \cdot \pr_x[\Bad(x) = t | f(x) = 1] \\
&\leq t^* + \sum_{t > t^* } \frac{t}{2^{t-1}p}\\
&= t^* + \frac{2}{p} \cdot \left(\frac{(t^* + 1)}{2^{t^*}} + \frac{1}{2^{t^*}}\right) = 2\log(1/p) + 2 \leq 2\log(2/p).
\end{align*}
\end{proof}

The rest of our argument is specific to $d=3$.
We restrict our attention to the accepting strings $x \in f^{-1}(1)$. For each
vertex $v \in V$ let $q(v) = \Pr_{x \sim f^{-1}(1)}[v \in
  \Path(x)]$. Each layer $t$ has three states $(t,1), (t,2)$ and
$(t,3) \in \Rej$. We assume that $q((t,1)) \geq q((t,2)) \geq q(t,3)
=0$ (since accepting strings never visit a $\Rej$ state).
We first bound the probability mass on states in the set $\Bad$.

\begin{Lem}
\label{lem:double-count}
We have
$$\sum_{v \in \Bad}q(v) = \E_{x \sim f^{-1}(1)}[\Bad(x)].$$
\end{Lem}
\begin{proof}
We have
\begin{align*}
\sum_{v \in \Bad}q(v) = \sum_{v \in \Bad}\Pr_{x \sim f^{-1}(1)}[x
  \text{ visits } v] 
= \E_{x \sim f^{-1}(1)}[\Bad(x)],
\end{align*}
by linearity of expectations.
\end{proof}

We partition the set $\Bad$ based on the value of $q(v)$:
$$\Bad^s = \left\{v \in \Bad: q(v) < \frac{1}{4}\right\}, 
\Bad^l = \left\{v \in \Bad: q(v) \geq \frac{1}{4}\right\}.$$
By Lemma \ref{lem:double-count} and Corollary \ref{cor:exp-bad} it
follows that $|\Bad^\ell| \leq 8\log(2/p)$. 

\begin{Lem}
We have
$$\Pr_{x \sim f^{-1}(1)}[\Path(x) \cap \Bad^s = \emptyset] \geq (p/2)^4.$$
\end{Lem}
\begin{proof}
Since for all $t$, $q((t,1)) \geq 1/2$ we have $(t,1) \not\in \Bad^s$.
Sort the vertices in $\Bad^s$ according to layer, so that $\Bad^s
=\{(t_1,2),\ldots,(t_w,2)\}$.  We have
$$\Pr_{x \sim f^{-1}(1)}[\Path(x) \cap \Bad^s = \emptyset] =
\prod_{i=1}^w\Pr_{x \sim f^{-1}(1)}[(t_i,2) \not\in
  \Path(x)|(t_1,2),\ldots,(t_{i-1},2) \not\in \Path(x)].$$  
Note that if  $(t_{i-1},2) \not\in \Path(x)$ then $(t_{i-1},1) \in
\Path(x)$. Hence conditioning on not visiting $(t_1,2),
\ldots,(t_{i-1},2)$ is the same as conditioning on visiting
$(t_1,1),\ldots,(t_{i-1},1)$. Further, conditioning on visiting $(t_1,1),\ldots,(t_{i-1},1)$ is the same as conditioning on $(t_{i-1},1)$. Therefore, 
\begin{align*}
\Pr_{x \sim f^{-1}(1)}[(t_i,2) \in \Path(x) |(t_1,1),\ldots,(t_{i-1},1) \in
  \Path(x)] & = \Pr_{x \sim f^{-1}(1)}[(t_i,2) \in \Path(x)
  |(t_{i-1},1) \in \Path(x)]\\
& \leq \frac{\Pr_{x \sim f^{-1}(1)}[(t_i,2) \in \Path(x)]}{\Pr_{x \sim f^{-1}(1)}[(t_{i-1},1) \in
    \Path(x)]}\\
&\leq \frac{q(t_i,2)}{q(t_{i-1},1)} \leq \frac{4}{3} \cdot q((t_i,2)),
\end{align*}
because $q(t_{i-1},1) = 1 - q(t_{i-1},2) \geq 3/4$. 
Hence we have
\begin{align*}
\Pr_{x \sim f^{-1}(1)}[\Path(x) \cap \Bad^s = \emptyset] 
& = \prod_{i=1}^w\Pr[(t_i,2) \not\in \Path(x)|(t_{i-1},1) \in
  \Path(x)]\\
& = \prod_{i=1}^w(1 - \frac{4q((t_i,2))}{3})  \geq e^{-2(\sum_{i=1}^w q((t_i,2)))} \geq (p/2)^{4}
\end{align*}
where we used the fact that for $z \leq 1/4$, $(1-4z/3) \geq e^{-2z}$ and $\sum_{v \in \Bad^s} q(v) \leq 2\log(2/p)$. 
\end{proof}

We are now ready to prove \tref{th:bpsdtobpint}.
\begin{proof}[Proof of \tref{th:bpsdtobpint}]
Observe that by the above claim, we can replace vertices in $\Bad^s$ by $\Rej$ vertices, and get a
new program $B'$ such that $B' \leq B$ and $\E[B'] \geq p \cdot (p/2)^4 \geq (p/2)^5$. Lastly, we handle the vertices in $\Bad^l$, which currently have transitions to \Rej. Assume that these vertices are $v_1,\ldots,v_j$ and that they read variables $x_{i_1},\ldots,x_{i_j}$. There exists a fixing $a_{i_1},\ldots,a_{i_j}$ of these variables such that the probability of acceptance of $B'$ over the remaining variables is at least $(p/2)^5$. Let $B'(a)$ denote the program obtained by hardwiring these values in $B'$. Now consider the program $B'' = B'(a) \wedge (x_{i_1} = a_{i_1})\wedge \cdots\wedge (x_{i_j} =a_{i_j})$, then $B'' \leq B'$ and 
$$\E[B''] \geq (p/2)^5 \cdot \frac{1}{2^{|\Bad^l|}} \geq (p/2)^{13},$$
since $|\Bad^l| \leq 8\log(2/p)$.

We only need to argue that $B'(a)$ and hence $B''$ is an intersection of width $2$ branching programs. Note that $B'(a)$ is a width $2$ program but with \Rej\ states for every vertex in $\Bad^s=\{(t_1,2),\ldots,(t_w,2)\}$. But we can view $B'(a)$ as an intersection of branching programs $B'_i$  for $i \in \{1,\ldots, w-1\}$, where $B'_i$ has start state $(t_i,1)$ and accept state $(t_{i+1},1)$. This completes the proof of the claim.
\end{proof}

\subsection{Reducing intersections of $\bp{2}$ to $\cnfx$}
We now perform the final step in our sequence of reductions to prove \tref{th:bpmain}.

\begin{Thm}\label{thm:bp2inttocnfx}
Let $f:\zo^n \rgta \zo$ be an intersection of width $2$ BPs on disjoint sets of inputs, i.e., $f = f_1 \wedge f_2 \wedge \cdots \wedge f_m$, where each $f_i \in \bp{2,n}$. Then, there exists a $\cnfx$ $g:\zo^n \rgta \zo$ such that, $g \leq f$ and $\E[g] \geq \E[f]^{O(1)}$. 
\end{Thm}

We use the following characterization of width $2$ branching programs
as decision lists due to Saks and Zuckerman \cite{SaksZ} and Bshouty,
Tamon and Wilson \cite{BshoutyTW98}.  For a set $S
\subseteq [n]$, let $\Ands{S}$ denote all functions of the form
$\wedge_jy_j$ where $j \in S$ and $y_j \in \{x_j, \bar{x_j}\}$. We
define $\Ors{S}$ and $\Xors{S}$ similarly. Note that all these classes
contain the constant functions.

\begin{Thm}[\cite{SaksZ,BshoutyTW98}]
\label{thm:sz}
Let $f \in \bp{2}$ be computed by a read-once, width $2$ branching program that reads
variables $x_S$ for $S \subseteq [n]$.  Then $f$ is computable by a
decision list $\mc{L}_f$ of the following form.
\begin{itemize}
\item $\mc{L}_f$ reads variables $x_V$ for some $V \subset S$ of size $k$.
\item There are $k+1$ leaves denoted $L_1,\ldots L_{k+1}$, where $L_j$ is labeled by a function $\ell_j \in   \Xors{S\setminus V}$\footnote{A decision list is a decision tree where the left child of every node is a leaf labeled by one of the functions $\ell_j$. On an input $x$, the output is computed by traversing the tree until a leaf is reached and outputting the value computed by the function at the leaf.}
\end{itemize}
\end{Thm}

We order $V$ according to how variables are read by $\mc{L}_f$ and use $V^j$ to 
denote the indices of the first $j$ variables. The condition that $x$ reaches
$L_j$ is given by a function in $g_j \in \Ands{V^j}$. We say that
$L_j$ accepts $x$ if $g_j(x) =1$ and $\ell_j(x) = 1$

We derive two consequences of Theorem \ref{thm:sz}.

\begin{Lem}\label{lm:sz1}
Let $f$ be as in \tref{thm:sz}. If $\E[f] \geq 5/6$, then there exists $g \in \Ors{V}$ such that $g \leq f$
and $\E[g] \geq \E[f]^9$.
\end{Lem}
\begin{proof}

Let $\E[f] = 1 -\eps$ for $\eps \leq \frac{1}{6}$. Note that
\begin{align*}
\eps = \sum_{j=1}^{k+1}2^{-j}\Pr[\ell_j(x) = 0]
\end{align*}

Consider the smallest $j$ such that $\ell_j$ is not the constant $1$ function. Since $\ell_j \in
\Xors{S\setminus V}$ and $\ell_j \neq 1$, $\ell_j$ rejects with probability
at least $1/2$, hence $\eps \geq 2^{-j -1}$.

The condition that $x$ reaches one of $L_1,\ldots,L_{j-1}$ is given by $g \in
\Ors{V^{j-1}}$. Since $\ell_1 \equiv \ell_2 \equiv \cdots \equiv \ell_{j-1}$, we have that $g \leq f$ and $\E[g] = 1- 2^{-j+1} \geq 1 - 4\eps$. Since $\eps \leq 1/6$, the inequality $(1- 4\eps) \geq
(1 - \eps)^9$ holds.
\end{proof}

\begin{Lem}\label{lm:sz2}
Let $f$ be as in \tref{thm:sz}. There exist $h_1 \in \Ands{V}$ and $h_2 \in \Xors{S\setminus V}$ such
that if we define $h = h_1\wedge h_2$ then $h \leq f$ and $\E[h] \geq \E[f]/3$.
\end{Lem}
\begin{proof}
Let $L_j$ be the highest leaf in $\mc{L}_f$ which is not labeled
$0$. Set $h_1 = g_j$ and $h_2 = \ell_j$. It is easy to see that 
\begin{align*}
\Pr_x[h(x) = 1]  = \Pr_x[L_j \ \text{accepts} \ x] \geq
  \frac{1}{3}\Pr_x[f(x) = 1].
\end{align*}
\end{proof}

We now prove \tref{thm:bp2inttocnfx}.
\begin{proof}[Proof of \tref{thm:bp2inttocnfx}]
Let $f = f_1 \wedge f_2 \wedge \cdots \wedge f_m$, where $f_i \in \bp{2,n}$. Let $p = \E[f]$. Then, for $I = \{i: \E[f_i] < 5/6\}$, $|I| < \log_{6/5}(1/p)$. For $i \notin I$, let $g_i$ be the function obtained from \lref{lm:sz1} and for $i \in I$, let $h_i$ be the function obtained from \lref{lm:sz2}. Let $g = \left(\wedge_{i \notin I} g_i\right) \wedge \left(\wedge_{i \in I} h_i \right)$. Then, clearly $g \in \cnfx$, $g \leq f$ and 
\[ \E[g] = \prod_{i \notin I} \E[g_i] \cdot \prod_{i \in I} \E[h_i] \geq \prod_{i \notin I} \E[f_i]^9 \cdot \prod_{i \in I} \left(\E[f_i]/3\right) \geq p^9 \cdot p \cdot \frac{1}{3^{|I|}} \geq p^{14}.\]  
\end{proof}

\subsection{\hsg\  for $\bp{3,n}$}
We now combine the previous sections to prove Theorems \ref{th:bpmain}, \ref{th:bphsg}.

\begin{proof}[Proof of \tref{th:bpmain}]
Follows immediately from combining Theorem \ref{th:bptosd}, \ref{th:bpsdtobpint}, \ref{thm:bp2inttocnfx}. 
\end{proof}

\begin{proof}[Proof of \tref{th:bphsg}]
Let $f \in \bp{3,n}$ with $\E[f] \geq \epsilon$. Let $g,k$ be as given by \tref{th:bpmain} applied to $f$ so that $\E[g] \geq \delta = (\epsilon/n)^c$. Let $G':\zo^s \rightarrow \zo^n$ be a \prg\ for $\cnfx$ with error at most $\delta/2$. By \tref{thm:prgcnfx}, there exists an explicit $G'$ with seed-length $s = O(\log (n/\epsilon) \cdot (\log \log (n/\epsilon))^3)$. 

Define, $G:\zo^{\log n + s} \rightarrow \zo^n$ as follows:
\begin{itemize}
\item Sample $r \sim [n]$ and $y \sim \zo^s$. 
\item Output $r$ $0$s followed by the first $n-r$ bits of $G'(y)$. 
\end{itemize}

We claim that $G$ is a $(\epsilon, (\epsilon/n)^{c+1})$-\hsg\  for $\bp{3,n}$. 

Assume that we guess $r = k$ correctly, which happens with probability
$1/n$. $G$ then simulates $g$ on the string $G'(y)$. Since, $\E[g] \geq \delta$, 
\[ \pr_{y \in \zo^s}\left[g(G'(y)) = 1\right] \geq \E[g] - \delta/2 \geq \delta/2.\]
Therefore, 
$$\Pr_{k,y \in \zo^s}\left[f(G(y)) =1\right] \geq \delta/2n.$$
The theorem now follows.
\end{proof}

\section{\prg s for Read-Once \cnf s}
\label{sec:prgcnf}

We construct a \prg\  for read-once \cnf s ($\rcnf$s) with a seed-length of $O((\log n)\cdot (\log \log n)^2)$ and error $1/\poly(n)$. As mentioned in the introduction, previously, only generators with seed-length $O(\log^2 n)$ were known for error $1/\poly(n)$. Besides being of interest on its own, this construction will play an important role in our \hsg\  for width 3 branching programs. Our main construction and its analysis are similar in spirit to what we saw for combinatorial rectangles and will be based on \tref{thm:main}.

\begin{Thm}\label{thm:prgcnf}
  For every $\epsilon > 0$, there exists an explicit \prg\  $G:\zo^r \rgta \dpm^n$ that fools all $\rcnf$s on $n$-variables with error at most $\epsilon$ and seed-length $r = O((\log (n/\epsilon)) \cdot (\log \log (n/\epsilon))^3)$. 
\end{Thm}

The core of our construction will be a structural lemma that can be summarized as follows: The bias function of a random restriction of $f$ where each variable has a small constant probability of being set has small $\Ll$-norm sandwiching approximators. 

Along with the structural lemma we shall also exploit the fact that for any $\rcnf$, randomly restricting a constant fraction of the inputs simplifies the formula significantly: with high probability a size $m$ $\rcnf$ upon a random restriction has size at most $\poly(\log n) \cdot m^\gamma$, where $\gamma < 1$ is a fixed constant. \tref{thm:prgcnf} is then proved using a recursive construction, where we use the above arguments for $O(\log \log n)$ steps.

\subsection{Sandwiching Approximators for Bias Functions}
For a function $f:\dpm^n \rgta [0,1]$, a subset $I \subseteq [n]$ and $x \in \dpm^I$, define $f_I(x):\dpm^I \rgta [0,1]$ by
\[ f_I(x) = \E_{y \in_u \dpm^{[n]\setminus I}}[ f(x \circ y)],\]
where $x \circ y$ denotes the appropriate concatenation: $(x \circ y)_i = x_i$ if $i \in I$ and $(x \circ y)_i = y_i$ if $i \notin I$. We call $f_I$ the ``bias function'' of the {\it restriction} $(x,I)$. 

We will show that for a $\rcnf$ $f$, and $I$ chosen in an almost $k$-wise independent manner, the bias function $f_I$ has small $\Ll$-norm sandwiching approximators with very high probability (over the choice of $I$).\ignore{We first define the notion of almost independent distributions.
\begin{Def}
  Let $0 < \alpha, \delta < 1/2$. We say a distribution on $\calD$ on $2^{[n]}$ is $\delta$-almost independent with bias $\alpha$ if $I \sim \calD$ satisfies the following conditions:
\begin{itemize}
\item For every $i \in [n]$, $\pr[i \in I] = \alpha$.
\item For any distinct indices $i_1,\ldots,i_k \in [n]$ and $b_1,\ldots,b_k \in \zo^k$, 
\[ \pr\left[\,\wedge_{j=1}^k (\mathsf{1}(i_j \in I) = b_j) \,\right] = \prod_{j=1}^k \pr[\mathsf{1}(i_j \in I) = b_j] \pm \delta.\]
\end{itemize}
\end{Def}
There exist explicit constructions of distributions in $\calD$ as above which only need $O(\log n + \log(1/\alpha \delta))$ random bits. We will write $I \sim \calD(\alpha,\delta)$ for short whenever $I$ is sampled from a $\delta$-almost independent distribution with bias $\alpha$ as above.

We are now ready to state our main structural lemma for $\rcnf$s. For a $\rcnf$ $f$, let $bias(f) = 1- \pr[f(x) = 0]$.}

\begin{Lem}[Main]\label{lm:maincnf}
There exists a constant $\alpha$ and $c > 0$ such that the following holds for every $\epsilon > 0$ and $\delta < (\epsilon/n)^c$. Let $f:\dpm^n \rgta \zo$ be a $\rcnf$ and $I \sim \calD(\alpha, \delta)$. Then, with probability at least $1 - \epsilon$, $f_I$ has $\epsilon$-sandwiching approximators with $\Ll$-norm at most 
$$L(n,\epsilon) = \left(n/\epsilon\right)^{c(\log \log (n/\epsilon))^2}.$$
\end{Lem}
\begin{proof}
Let $f = C_1 \wedge C_2 \wedge \cdots \wedge C_m$. By abuse of notation, we will let $C_i$ denote the set of variables appearing in $C_i$ as well. In our analysis we shall group the clauses based on their widths. Let $\beta = 1+1/6$.

We first handle the case where $f$ has bias at least $\epsilon$, i.e., $\pr[f(x) = 1] \geq \epsilon$. Let $W_\ell = c_1 \log \log(n/\epsilon)$ and $W_u = \log_2 m$ for $c_1$ a constant to be chosen later. Let $f_\ell$ be the $\rcnf$ containing all clauses of width less than $W_\ell$ and $f_u$ the $\rcnf$ containing all clauses of width at least $W_u$. Let $T = \log_\beta (W_u/2W_\ell)$. For $w \in W_B \equiv \{\lfloor W_\ell \beta^r \rfloor: 0 \leq r \leq T\}$, let $f_w$ be the $\rcnf$ containing all clauses with width in $[w,\beta w)$. Then,
\begin{equation}
  \label{eq:widthbucket}
  f \equiv f_\ell \,\wedge\, \left(\wedge_{w \in W_B} f_w \right) \,\wedge f_u. 
\end{equation}

We will show that each of the functions $f_\ell, f_w, f_u$ have good sandwiching approximators. We then use \tref{thm:xor} to conclude that $f$ has good sandwiching approximators. The claim for $f_\ell$ follows immediately from \tref{th:smallcnf}. The main challenge will be in analyzing $f_w$ (the analysis for $f_u$ is similar). To show that $f_w$ has good sandwiching approximators, we shall appeal to \tref{thm:corsym}.

Observe that as $f$ has bias at least $\epsilon$, the number of clauses of width at most $w$ in $f$ is at most $2^w \log(1/\epsilon)$. We will repeatedly use this fact. Let $\epsilon_1 = \epsilon/\poly(n,1/\epsilon)$ to be chosen later.  

\paragraph{Sandwiching $f_\ell$.} As each clause in $f_\ell$ has width at most $W_\ell$, the number of clauses in $f_\ell$ is at most $m_\ell \leq 2^{W_\ell} \log(1/\epsilon)\leq (\log(n/\epsilon))^{c_1+1}$. Thus, by \tref{th:smallcnf}, $f_\ell$ has $\epsilon_1$-sandwiching approximators of $\Ll$-norm at most $m_\ell^{O(\log(1/\epsilon_1))} = (\log(n/\epsilon))^{O(\log(1/\epsilon_1))}$. As $\Ll$-norm does not increase under averaging over a subset of the variables, it follows that $(f_\ell)_I$ has $\epsilon_1$-sandwiching polynomials with the same $\Ll$ norm bound:
\begin{equation}
  \label{eq:cnfsmall}
  (f_\ell)_I\text{ has $\epsilon_1$-sandwiching approximators with $\Ll$-norm at most $(\log(n/\epsilon))^{O(\log(1/\epsilon_1))}$.}
\end{equation}

\paragraph{Sandwiching $f_w$.} Fix a $w \in W_B$.  Note that $f_w$ has $m_w < 2^{\beta w} \log(1/\epsilon)$ clauses. Without loss of generality, suppose that $f_w = C_1 \wedge C_2 \wedge \cdots \wedge C_{m_w}$. Let $I \sim \calD(\alpha,\delta)$.

Let $J \subseteq [m_w]$ be the set of all {\it good} clauses, $J = \{j: |C_j \cap I| \leq w/3\}$. Decompose $f_w = f_w' \wedge f_w''$, where $f_w' = \wedge_{j \in J} C_j$. We first show that $(f_w')_I$ has good sandwiching approximators. We then show that $f_w''$ has a small number, $\poly(\log(n/\epsilon))$, of clauses with high probability over $I$. The intuition for the first step is that if each $|C_j \cap I|$ is small, then the randomness in the remaining variables damps the variance of the bias function $f_I$ enough to guarantee existence of good sandwiching approximators via \tref{thm:corsym}. For the second step, intuitively, as $I$ picks each element with probability at most $\alpha$, we expect $|C_j \cap I|$ to be about $\alpha |C_j| < \alpha (\beta w) \ll w/3$. Thus, the probability that $|C_j \cap I|$ is more than $w/3$ should be small so that the total number of bad clauses is small with high probability.

For brevity, suppose that $f'_w = C_1 \wedge \cdots \wedge C_{m'}$ and let $w_j = |C_j| \in [w,\beta w)$. For $x \in \dpm^I$, and $j \in [m']$, define $g_j':\dpm^I \rgta [-1,1]$ by 
\[ g_j'(x) =
\begin{cases}
  -1/2^{w_j} &\text{if $x$ satisfies $C_j \cap I$}\\
1/2^{|C_j \setminus I|}-1/2^{w_j} &\text{otherwise}
\end{cases}.
\]
Then, for $p_j = 1 - 1/2^{w_j}$,
\begin{align*}
  (f_w')_I(x) &= \prod_{j \in J} \left((1-1/2^{w_j}) - g_j'(x)\right) = \prod_{j \in J}\left(p_j - g_j'(x)\right) = \left(\prod_{j \in J} p_j\right)\cdot \prod_{j \in J}\left(1 - \frac{g_j'(x)}{p_j}\right).
\end{align*}
Let $g_j(x) = g_j'(x)/p_j$. Then, 
\[ (f_w')_I(x) = \prod_{j \in J} p_j \cdot \prod_{j \in J}(1 - g_j(x)).\]

By expanding the above expression we can write $(f_w')_I(x) = \sum_{k=1}^{m'} c_k S_k(g_1,\ldots,g_{m'})$ where the coefficients $c_k$ are at most $1$ in absolute value. We will show that $g_1,\ldots,g_{m'}$ satisfy the conditions of \tref{thm:corsym}.

Clearly, $g_1,\ldots,g_{m'}$ are on disjoint subsets of $x$. Note that $g_j'(x) \in [-1/2^{2w/3}, 1/2^{2w/3}]$. Hence, as $p_j \geq 1/2$, $g_j(x) \in [-\sigma, \sigma]$ for $\sigma = 2/2^{2w/3}$. Now, as $w \geq c_1\log\log(1/\epsilon)$, $m' \leq 2^{\beta w}\log(1/\epsilon) \leq 2^{(\beta + 1/c_1)w}$. Thus, for $c_1 > 12$,
\[ \sigma = \frac{2}{2^{2w/3}} \leq \frac{1}{(m')^{1/2+1/12}}.\]
Finally, note that each $g_j$ has $\Ll$-norm at most $2$. This is because any clause, and hence $g_j'$, has $\Ll$-norm at most $1$. Therefore, the functions $g_j$ satisfy the conditions of \tref{thm:corsym}. Thus, 
\begin{equation}
  \label{eq:cnfw1}
(f_w')_I\text{ has $\epsilon_1$-sandwiching approximators with $\Ll$-norm at most $\poly(1/\epsilon_1)$.}
\end{equation}

\eat{Thus, $\E[g_j^k] \leq \sigma^k$. Further, $\E[g_j] = 0$ and  each $g_j$ is a bounded function in $[-1,1]$ depending on at most $\beta w$ variables, $\Ll(g_j) \leq 2^{\beta w}$.  Finally, as $m' \leq 2^{\beta w} \log(1/\epsilon)$,
\[ \frac{1}{(m')^{1.1} \log(1/\epsilon)^{100}} \geq \frac{1}{2^{1.1 \beta w} \log(1/\epsilon)^{102}} \geq \frac{1}{2^{(1.1 \beta + 102/c_1) m}} \geq \frac{4}{2^{4w/3}} = \sigma^2,\]
where we used the fact that $w > c_1 \log \log (1/\epsilon)$ and $c_1$ is a big enough constant.

Thus, the functions $g_j$, $j \leq m'$ satisfy the assumptions of \tref{thm:main}. Hence, by \tref{thm:main},
\begin{equation}
  \label{eq:cnfw1}
(f_w')_I\text{ has $\epsilon_1$-sandwiching approximators with $\Ll$-norm at most $(\log(1/\epsilon))^{O(\log(1/\epsilon_1))}$.}  
\end{equation}}

We are almost done, but for $(f_w'')$. We will show that with high probability over $I$, $(f_w'')$ has $O(\log(n/\epsilon))$ clauses. To do so we will follow a standard argument for showing large deviation bounds using bounded independence.

For $i \in [w]$ and $j \in [m_w]$, let $X_{ij}$ be the indicator variable that is $1$ if the variable corresponding to the $i$'th literal in the $j$'th clause of $f_w$ is included in $I$ and $0$ otherwise. Let 
$$X = S_k\left(\,S_{w/3}(X_{11}, X_{21},\ldots, X_{w1}), \ldots, S_{w/3}(X_{1m_w}, X_{2m_w}, \ldots, X_{wm_w})\,\right).$$
Then, for any $k$, 
$$\pr[size(f_w'') \geq k] \leq \E[X].$$
 To see this observe that whenever $size(f_w'') \geq k$, $X$ is at least $1$. Let us first calculate this expectation when the variables $X_{ij}$ are truly independent. In this case, as $m_w \leq 2^w \log(1/\epsilon)$, and each clause has at most $\beta w$ variables, 
\begin{align*}
 \E[X] &= \binom{m_w}{k} \cdot \binom{\beta w}{w/3} \cdot \alpha^{wk/3} \\
&\leq 2^w \cdot \left(\frac{\beta e \log(1/\epsilon)}{k}\right)^k \cdot (8 \alpha)^{wk/3}.
\end{align*}

Therefore, for $\alpha = 1/32$ and $k = c_2 \max(\log(n/\epsilon)/w,1)$ for $c_2$ sufficiently large, $\E[X] \leq \epsilon/2n$. Now, as the actual variables $X_{ij}$ are $\delta$-almost independent, and the polynomial defining $X$ has at most $\binom{m_w}{k} \cdot \binom{\beta w}{w/3}$ terms, the expectation for $I \sim \calD$ can be bounded by
\[ \E[X] \leq \frac{\epsilon}{2\log n} + \delta \cdot \binom{m_w}{k} \cdot \binom{\beta w}{w/3} = \frac{\epsilon}{2n} + \delta \cdot 2^{O(wk)} \leq \frac{\epsilon}{n},\]
for $\delta \leq (\epsilon/n)^c$ for $c$ a sufficiently large constant. Combining the above equations and applying \tref{th:smallcnf}, we get that with probability at least $1-\epsilon/ n$, $(f_w'')$ has $\epsilon_1$-sandwiching polynomials with $\Ll$-norm at most 
$$k^{O(\log(1/\epsilon_1))} = (\log(n/\epsilon))^{O(\log(1/\epsilon_1))}.$$
Therefore, from \eref{eq:cnfw1} and \tref{thm:xor}, with probability at least $1 - \epsilon/n$,  
\begin{equation}
  \label{eq:cnfw}
(f_w)_I\text{ has $O(\epsilon_1)$-sandwiching approximators with $\Ll$-norm at most $(\log(n/\epsilon))^{O(\log(1/\epsilon_1))}$.}  
\end{equation}

\paragraph{Sandwiching $f_u$:} A careful examination of the argument for $f_w$ reveals that we used two main properties: there are at most $2^{\beta w} \log(1/\epsilon)$ clauses in $f_w$ and every clause has length at least $w$. Both of these are trivially true for $f_u$ with $w = W_u$. Thus, the same argument applies. In particular, with probability at least $1 - \epsilon/n$, 
\begin{equation}
  \label{eq:cnflarge}  
(f_u)_I\text{ has $O(\epsilon_1)$-sandwiching approximators with $\Ll$-norm at most $(\log(n/\epsilon))^{O(\log(1/\epsilon_1))}$.}  
\end{equation}

Now, observe that 
$$f_I(x) = (f_\ell)_I(x) \cdot  (f_u)_I(x) \cdot \prod_{w \in W_B}(f_w)_I(x).$$
Therefore, we can apply \tref{thm:xor}. In particular, by Equations \ref{eq:cnfsmall}, \ref{eq:cnfw}, \ref{eq:cnflarge}, and a union bound, for $b = |W_B| + 2 = O(\log \log n)$ we have: with probability at least $1 - \epsilon$, $f_I$ has $(16^b \epsilon_1)$-sandwiching polynomials with $\Ll$-norm at most 
$$4^{b}\cdot \left(\left(\log(n/\epsilon)\right)^{O(b\,\log(1/\epsilon_1))}\right) = \left(1/\epsilon_1\right)^{O((\log \log (n/\epsilon))^2)}.$$
The lemma now follows by setting $\epsilon_1 = \epsilon/n^{O(1)}$.

\paragraph{Handling  Small Bias Case.} We now remove the assumption that $\E[f] \geq \epsilon$. Suppose $\E[f] \leq \epsilon$. Consider the formula $f'$ obtained from $f$ by removing clauses in $f$ until the first time $\E[f']$ exceeds $\epsilon$. Then, $f \leq f'$ and $\epsilon \leq \E[f'] \leq 2\epsilon$ (as each clause has probability at most $1/2$ of being false). We can use the upper approximator for $f'$ as an upper approximator for $f$ and constant zero as a lower approximator. This completes the proof of lemma.
\end{proof}

\subsection{Restrictions Simplify $\rcnf$s}
We next argue that for restrictions $(x,I)$ where $(x,I)$ are chosen from almost-independent distributions as in the previous section, $\rcnf$s simplify significantly and in particular have few surviving clauses with very high probability. 

Let $I \sim \calD(\alpha,\delta)$ be as in \lref{lm:maincnf} and $x \sim \calD$ be chosen from a $\delta_1$-biased distribution with $\delta_1 = 1/\poly(n)$. We will show that fixing the variables in $I$ according to $x$ will make the number of clauses drop polynomially. Let $\alpha,\beta$ be the constants from \lref{lm:maincnf}. 
\begin{Lem}\label{lm:shrink}
There exists constants $c_2, \gamma > 0$ such that the following holds for $\delta, \delta_1 < (\epsilon/n)^{c_2}$. Let $I \sim \calD(\alpha,\delta)$ and $x \sim \calD$ where $\calD$ is a $\delta_1$-biased distribution on $\dpm^n$. Let $f:\dpm^n \rgta \zo$ be a $\rcnf$ with $\E[f] \geq \epsilon$. Let $g:\dpm^{[n] \setminus I}\rgta \zo$ be the $\rcnf$ obtained from $f$ by fixing the variables in $I$ to $x$. Then, with probability at least $1- \epsilon$ over the choice of $(x, I)$, $g$ is a $\rcnf$ with at most $(\log (n/\epsilon))^{c_2} \cdot m^{1-\gamma}$ clauses.
\end{Lem}
\begin{proof}
As in the proof of \lref{lm:maincnf}, we shall do a case analysis based on the width of the clauses. Let $f_\ell, f_w, f_u$ and $W_B$ be as in \eref{eq:widthbucket}. Note that the number of clauses in $f_\ell$ is at most $2^{W_\ell} \log(1/\epsilon) = \poly(\log(n/\epsilon))$. We will now reason about each of the $f_w$'s for $w \in W_B$. The argument for $f_u$ is similar and is omitted. 

Let $f_w$ have $m_w$ clauses, where $m_w > 8 \log(1/\epsilon)$, otherwise there is nothing to prove. Without loss of generality, suppose that $f_w = C_1 \wedge C_2 \wedge \cdots \wedge C_{m_w}$ and $w_j = |C_j|$. Let $Y_j$ be the indicator variable that is $1$ if $C_j$ survives in $g$ (i.e., is not fixed to be {\it true}) and $0$ otherwise. We first do the calculations assuming that the variables in $x$ and $I$ are truly independent and later transfer these bounds to the almost independent case. 

Observe that $\pr[Y_j = 1] = (1 - \alpha/2)^{w_j} \leq (1 - \alpha/2)^{w}$.  Let $M, k < M$ be parameters to be chosen later. Then, as in the proof of \lref{lm:maincnf}, 
\begin{align*}
 \pr\left[\,\sum_j Y_j > M\,\right] \cdot \binom{M}{k} \leq \E\left[ S_k(Y_1,\ldots,Y_{m_w})\right]  &\leq \binom{m_w}{k} \cdot \left(1 - \frac{\alpha}{2}\right)^{wk}.
\end{align*}
Here, the first inequality follows from observing that if $\sum_j Y_j > M$, then $S_k(Y_1,\ldots,Y_{m_2})$ is at least $\binom{M}{k}$. Therefore,
\[ \pr\left[\, \sum_j Y_j > M\,\right] \leq \left(\frac{m_w e}{M}\right)^k \cdot \left(1 - \frac{\alpha}{2}\right)^{wk}.\]
Now, setting $M = m_w^{1 - \gamma} (e \log(1/\epsilon))$ for a sufficiently small constant $\gamma$ and using the fact that $m_w < 2^{\beta w} \log(1/\epsilon)$, it follows that 
\[ \pr[\, \sum_j Y_j > M\,] \leq  \left(\frac{(2-\alpha) 2^{\beta \gamma}}{2}\right)^{wk} < 2^{-\Omega(wk)},\]
for $\gamma$ a sufficiently small constant. 
Thus, for $k = c_3 \max(\log(n/\epsilon)/w, 1)$ and $c_3$ a sufficiently large constant, $\pr[\, \sum_j Y_j > M\,] < \epsilon/2n$. Now, as in the proof of \lref{lm:maincnf}, transferring the above calculations to the case of almost independent distributions only incurs an additional error of
\[ err = (\delta + \delta_1) \cdot \binom{m_w}{k} \cdot 2^{\beta w} = (\delta + \delta_1) \cdot \poly(n,1/\epsilon).\]
Therefore, for $\delta, \delta_1 < (\epsilon/n)^{c'}$ for a sufficiently large constant $c'$, we get $\pr[\, \sum_j Y_j > M\,] < \epsilon/n$. 

Hence, by a union bound over $w \geq W_\ell$, with probability at least $1 - \epsilon$, the number of surviving clauses in $g$ is at most
\begin{multline*}
 size(f_\ell) + (e \log(1/\epsilon)) \cdot \sum_{w \in W_B} m_w^{1-\gamma} \leq \poly(\log n) + (e \log(1/\epsilon)) \cdot |B|^\gamma \cdot (\sum_w m_w)^{1-\gamma} < \\ 
\poly(\log n) + (e\log(1/\epsilon)) \cdot |B|^\gamma \cdot m^{1-\gamma},  
\end{multline*}

where the first inequality follows from the power-mean inequality. The claim now follows.
\end{proof}

In our recursive analysis we will also have to handle $\rcnf$s that need not have high acceptance probabilities. The following corollary will help us do this.
\begin{Cor}\label{cor:shrinkgen}
Let constants $c_2,\gamma$ and $\delta,\delta_1, I \sim \calD(\alpha,\delta), x \sim \calD$ be as in \lref{lm:shrink}. Let $f:\dpm^n \rgta \zo$ be a $\rcnf$, and let $g:\dpm^{[n] \setminus I}\rgta \zo$ be the $\rcnf$ obtained from $f$ by fixing the variables in $I$ to $x$. Then, with probability at least $1- \epsilon$ over the choice of $(x, I)$, there exist two $\rcnf$s $g_\ell,g_u$ of size at most $(\log (n/\epsilon))^{c_2} \cdot m^{1-\gamma}$ such that $g_\ell \leq g \leq g_u$ and $\E[g_u] - \E[g_\ell] \leq \epsilon$.
\end{Cor}
\begin{proof}
  If $\E[f] \geq \epsilon/2$, the claim follows from \lref{lm:shrink}. Suppose $\E[f] \leq \epsilon/2$. Let $f'$ be the formula obtained from $f$ by throwing away clauses until $\E[f']$ exceeds $\epsilon/2$. Then, $\epsilon/2 \leq \E[f'] \leq \epsilon$. Let $g'$ be the $\rcnf$ obtained from $f'$ by restricting the variables in $I$ to $x$. The claim now follow by applying \lref{lm:shrink} to $f'$ and setting $g_\ell \equiv 0$ and $g_u \equiv g'$.
\end{proof}

\subsection{A Recursive \prg\ Construction for $\rcnf$s}
We now use Lemmas \ref{lm:maincnf} and \ref{lm:shrink} recursively to prove \tref{thm:prgcnf}. The main intuition is as follows. 

Let $\epsilon = 1/\poly(n)$. \lref{lm:maincnf} ensures that with high probability over the choice of $I$, $f_I$ is fooled by small-bias spaces with bias $n^{-O((\log \log n)^2)}$ which can be sampled from using $O((\log n)(\log \log n)^2)$ random bits. Note that $I$ can be sampled using $O(\log n)$ random bits. 

Consider any fixed $I \subseteq [n]$ and $x \in \dpm^I$. We wish to apply the same argument to $f_{(x,I)}:\dpm^{[n]\setminus I}\rgta \zo$ to pick another set $I_1 \subseteq [n]$ and $x_1 \in \dpm^{I_1}$ and so on. The saving factor will be that most of the clauses in $f$ will be determined by the assignment to $x$. In particular, by \lref{lm:shrink}, with probability $1 - 1/\poly(n)$, $f_{(x,I)}$ has at most $\tilde{O}(n^{1 - \gamma})$ clauses. By repeating this argument for $t = O_\gamma(\log \log n)$ steps we will get a $\rcnf$ with at most $\poly(\log n)$ clauses, which can be fooled directly. The total number of random bits used in this process will be $O((\log n)(\log \log n)^3)$.

Fix $\epsilon > 0$ and let constants $\alpha, c$ be as in \lref{lm:maincnf}. Let $\calD(\alpha,\delta)$ be a $\delta$-almost independent distribution on $2^{[n]}$ with bias $\alpha$. Finally, let $\calD(\delta_1), \calD(\delta_2)$ denote $\delta_1$-biased and $\delta_2$-biased distributions on $\dpm^n$ respectively for $\delta_1,\delta_2$ to be chosen later. Let $T = C \log \log n$ for $C$ to be chosen later. Consider the following randomized algorithm for generating a string $z \in \dpm^n$.
\begin{itemize}
\item For $t = 1,\ldots,T$, generate independent samples $z^1,\ldots,z^T \sim \calD(\delta_1)$ and $J_1,\ldots,J_T \sim \calD(\alpha,\delta)$. 
\item Let $I_1 = J_1$ and $I_t = J_t\setminus \left(\cup_{r=1}^{t-1} I_r\right)$ for $2 \leq t \leq T$. This is equivalent to sampling $I_t$ from a $\delta$-almost independent distribution with bias $\alpha$ from the set of subsets of as yet ``uncovered'' elements $[n] \setminus \cup_{r=1}^{t-1} I_r$. 
\item Let $x^t = (z^t)_{I_t}$. This is equivalent to sampling $x^t$ using a $\delta_1$-biased distribution over $\dpm^{I_t}$. 
\item Let $I = \cup_{t = 1}^T I_t$ and $x = x^1 \circ x^2 \circ \cdots \circ x^T \in \dpm^{I}$ be the appropriate concatenation: for $i \in I$, $(x_i) = (x^t)_i$ if $i \in I_t$.
\item Let $y \sim \calD(\delta_2)$. The final generator output is defined by
  \begin{equation}
    \label{eq:gencnf}
G(z^1,\ldots,z^T,J_1,\ldots,J_T,y) = z, \text{ where $z_i = x_i$ if $i \in I$ and $z_i = y_i$ otherwise.}    
  \end{equation}
\end{itemize}

To analyze our generator we first show that the restriction $(x,I)$ preserves the bias of $\rcnf$s. Let $L(n,\epsilon)$ be the bound from \lref{lm:maincnf}.
\begin{Lem}\label{lm:preservebias}
For $x,I$ defined as above, with probability at least $1 - \epsilon\, T$ over the choice of $I$, for every $\rcnf$ $f:\dpm^n \rgta \zo$, 
$$\left|\E_x[f_I(x)] - \E_{y \in_u \dpm^I}[f_I(y)]\right| < \delta_1\cdot L(n,\epsilon) \cdot T + 2\,\epsilon\,T.$$
\end{Lem}
\begin{proof}
\newcommand{\djo}{\D^{j-1}}
\newcommand{\djj}{\D^j}
We will prove the claim by a hybrid argument. For $j \leq T$, let $y^j \sim \dpm^{I_j}$ and let $\D^j$ denote the distribution of $x^1 \circ x^2 \circ \cdots \circ x^{j} \circ y^{j+1} \circ \cdots \circ y^T$(the concatenation is done as in the definition of $x$). Note that $\djo$ and $\djj$ differ only in the $j$'th concatenation element, $x^j, y^j$. Further, $\D^0$ is uniformly distributed on $\dpm^I$ and $\D^T$ is the distribution of $x$. We will show that with probability at least $1-\epsilon$, over the choice of $I$,
\[ \left|\E_{a \sim \djo}\left[f_I(a)\right] - \E_{a \sim \djj}\left[f_I(a)\right]\right| < \delta_1 \cdot L(n,\epsilon),\]

We couple the distributions $\djo$ and $\djj$ by drawing $x^i$ for $i < j$ and let $I^j = \cup_{r \leq j} I_r$. Now, as $y^{j+1},\ldots,y^T$ are chosen uniformly at random,
\begin{align*}
 \E_{a \sim \djo}[f_I(a)] &= \E\left[f_{I^j}(x^1 \circ \cdots \circ x^{j-1} \circ x^j)\right],\\
\E_{a \sim \djj}[f_I(a)] &= \E\left[f_{I^j}(x^1 \circ \cdots \circ x^{j-1} \circ y^j)\right].
\end{align*}

Consider any fixing of the variables $x^1,\ldots,x^{j-1}$ and $I_1,\ldots,I_{j-1}$ and let $g:\dpm^{[n] \setminus \cup_{r < j} I_r} \rgta \zo$ be the $\rcnf$ obtained from $f$ under this fixing. Then, by \lref{lm:maincnf}, $g_{I_j}$ is fooled by small-bias spaces: with probability $1-\epsilon$ over the choice of $I_j$,
\[ \left|\E_{x^j}\left[g_{I_j}(x^j)\right] - \E_{y^j}\left[g_{I_j}(y^j)\right] \right| \leq \delta_1 \cdot L(n,\epsilon).\]

Combining the above three equations, we have with probability at least $1-\epsilon$ over $I_j$, 
\begin{align*}
  \left|\E_{a \sim \djo}\left[f_I(a)\right] - \E_{a \sim \djj}\left[f_I(a)\right]\right| &= \left|\E\left[f_{I^j}(x^1 \circ \cdots \circ x^{j-1} \circ x^j)\right] - \E\left[f_{I^j}(x^1 \circ \cdots \circ x^{j-1} \circ y^j)\right] \right|\\
&= \left| \E_{x^1,\ldots,x^{j-1}} \E_{x^j}\left [g_{I_j}(x^j)\right] - \E_{x^1,\ldots,x^{j-1}} \E_{y^j}\left[g_{I_j}(y^j)\right] \right|\\
&\leq \E_{x^1,\ldots,x^{j-1}} \left[\, \left|\E_{x^j}\left [g_{I_j}(x^j)\right] - \E_{y^j}\left[g_{I_j}(y^j)\right] \right|\,\right] \leq \delta_1 \cdot L(n,\epsilon).
\end{align*}
The claim now follows by taking a union bound for $j = 1,\ldots, T$.
\end{proof}

We are now ready to prove our main \prg\  construction. The idea is to combine Lemmas \ref{lm:shrink}, \ref{lm:preservebias}. For $(x,I)$ chosen as in \lref{lm:preservebias} we do not change the bias of the restricted function, on the other hand by iteratively applying \ref{lm:shrink} we can show that the resulting restricted $\rcnf$ has $(\log n)^{O(\log \log n)}$ clauses and hence is fooled by $n^{-O((\log \log n)^2)}$-biased distributions. 

\begin{proof}[Proof of \tref{thm:prgcnf}]
Let $I,x,y,z$ be as defined in \eref{eq:gencnf}. Fix a $\rcnf$ $f:\dpm^n \rgta \zo$. Let $g:\dpm^{[n]\setminus I} \rgta \zo$ be the $\rcnf$ obtained by from $f$ by fixing the variables in $I$ to $x$. Let $I' = [n] \setminus I$. Note that
\begin{equation}
  \label{eq:prgcnf1}
 f_I(x) = \E_{y' \sim \dpm^{I'}}\left[ g(y')\right].  
\end{equation}

We next argue that $g$ is fooled by small-bias spaces with high probability over the choice of $x,I$. Observe that $g$ can be viewed as obtained from $f$ by iteratively restricting $f$ according to $(x^1,I_1), (x^2,I_2),\ldots,(x^T,I_T)$ and all of these are independent of one another. Therefore, by \cref{cor:shrinkgen} and a union bound, with probability at least $1 - \epsilon\cdot T$, $g$ has $O(\epsilon T)$-sandwiching $\rcnf$s $g_\ell,g_u$ of size at most 
$$M = (\log (n/\epsilon))^{c_2 T} \cdot m^{(1-\gamma)^T} = (\log (n/\epsilon))^{O(\log \log n)},$$ 
for $T = C\log \log n$ and $C$ a large constant. Hence, by \tref{th:smallcnf}, $g_\ell,g_u$ are $\epsilon$-fooled by $\delta_2$-biased distributions for $\delta_2 = M^{-O(\log(1/\epsilon))}$. As $g_\ell,g_u$ sandwich $g$, it follows that $g$ is $O(\epsilon T)$-fooled by $\delta_2$-biased distributions. As the above is true with probability at least $ 1- \epsilon \, T$ over the choice of $(x,I)$, by taking expectation over $(x,I)$ we get ($y$ is $\delta_2$-biased)
\begin{equation}
  \label{eq:prgcnf2}
 \E_{x,I}\left[\, \left|\E_y\left[g(y)\right] - \E_{y'\sim \dpm^{I'}}\left[ g(y')\right]\right |\,\right] = O(\epsilon T).  
\end{equation}

Combining Equations \ref{eq:prgcnf1}, \ref{eq:prgcnf2}, we get
\begin{align*}
  \label{eq:cnfp1}
  \pr[f(z) = 1] = \E_{x,I}\left[ \E_y[g(y)]\right] &= \E_{x,I}\left[ \E_{y'\sim \dpm^{I'}}\left[ g(y')\right]\right] \pm O(\epsilon T)\nonumber\\
&= \E_{x,I}\left[f_I(x)\right] \pm O(\epsilon T).
\end{align*}

Finally, note that for any $I \subseteq [n]$, 
\[ \pr_{z' \sim \dpm^n}[f(z') = 1] = \E_{x' \sim \dpm^I}[f_I(x')].\]

Combining the above two equations with \lref{lm:preservebias}, we get 
\begin{align*}
 \left| \pr[f(z) = 1] - \pr_{z' \sim \dpm^n}[f(z') = 1] \right| &\leq \left|\E_{x,I}\left[f_I(x)\right] - \E_{x' \sim \dpm^I}[f_I(x')]\right| + O(\epsilon T) \\
&\leq \delta_1 \cdot L(n,\epsilon) + O(\epsilon T). 
\end{align*} 

Therefore, by setting $\delta_1 = \epsilon/L$ the above error is at most $O(\epsilon  T)$. The number of bits used by the generator is 
\begin{align*}
  T \text{ (bits needed for $x^1, I_1$)} + \text{ (bits needed for y)} &= T \cdot O(\log n + \log(1/\delta_1)) + O(\log n + \log(1/\delta_2)) \\
&= O((\log (n/\epsilon)) \cdot (\log \log (n/\epsilon))^3).
\end{align*}
The theorem now follows by rescaling $\epsilon = \epsilon'/c'(\log \log n)$ for a large constant $c'$.
\end{proof}


\section{A \prg\ for $\cnfx$}
\label{sec:cnfx}
 
We construct a \prg\  for the class of $\cnfx$. The generator will be the same as in \tref{thm:prgcnf}. The analysis will also be similar and in fact follow easily from \tref{thm:prgcnf}. To do this, we shall use the following simple claim.

\begin{Lem}\label{lm:andparity}
  Let $f:\dpm^n \rgta \zo$ be a conjunction of parity constraints on $n$ variables. Then, $f$ has $\Ll$-norm at most $1$.
\end{Lem}
\begin{proof}
  Let $S_1,S_2,\ldots,S_m$ be the subsets defining the parity constraints in $f$. Then,
\[ f(x) = \prod_{j=1}^m \,\left(\frac{1 - \prod_{i \in S_j} x_i}{2}\right).\]
The lemma now follows.
\end{proof}

\begin{Thm}\label{thm:prgcnfx}
  For every $\epsilon > 0$, there exists an explicit \prg\ $G:\zo^r \rgta \dpm^n$ that fools all $\cnfx$ formulas on $n$-variables with error at most $\epsilon$ and seed-length $r = O((\log (n/\epsilon)) \cdot (\log \log (n/\epsilon))^3)$. 
\end{Thm}
\begin{proof}
Let $G$ be the generator from \tref{thm:prgcnf}. We will show that $G$ fools $\cnfx$ as well. This does not follow in a black-box manner from \tref{thm:prgcnf}, but we will show analogues of \tref{th:smallcnf}, \lref{lm:maincnf} and \lref{lm:shrink} hold so that the rest of the proof of \tref{thm:prgcnf} can be used as is.
  
Let $f:\dpm^n \rgta \zo$ be a $\cnfx$ of size $m$. Let $f = g \wedge h$, where $g$ has all the parity constraints of $f$ and $h$ the clauses. 

First observe that by \lref{lm:andparity} and \tref{th:smallcnf}, a similar statement holds for $f$. Let $P_\ell, P_u$ be the $\epsilon$-sandwiching approximators for $h$ as guaranteed by \tref{th:smallcnf}. Then, $P_\ell' := g \cdot P_\ell$, $P_u' = g \cdot P_u$ are $\epsilon$-sandwiching approximators for $f$ and the $\Ll$-norm of $P_\ell'$ ($P_u'$) is bounded by the $\Ll$-norm of $P_\ell$ ($P_u$) by \lref{lm:andparity}.

 Note that for any subset $I \subseteq [n]$, $g_I:\dpm^I \rgta \zo$ is a constant function. Therefore, $f_I(x) = c_I \cdot h_I(x)$, where $c_I \leq 1$. Thus, by applying \lref{lm:maincnf} to $h$, we get an analogous statement for $f$. 

Finally, we show an analogue of \lref{lm:shrink}. Suppose that $f$ has acceptance probability at least $\epsilon$. Then, $g$ has at most $\log_2(1/\epsilon)$ clauses. Therefore, by \lref{lm:shrink} applied to $h$, we also get a similar statement for $f$ with a slightly worse constant of $c_2' = c_2 + 1$. By arguing as in the proof of \cref{cor:shrinkgen}, we get a similar statement for $f$.

Examining the proof of \lref{thm:prgcnf} shows that given the above
analogues of \tref{th:smallcnf}, \lref{lm:maincnf} and
\lref{lm:shrink}, the rest of the proof goes through. The theorem
follows. 
\end{proof}

\bibliographystyle{amsalpha}
\bibliography{references}

\end{document}